\pdfoutput=1
\documentclass[11pt]{article}

\usepackage[a4paper,margin=2.6cm]{geometry}
\usepackage{amsmath,amssymb,mathtools,amsthm}
\newtheorem{proposition}{Proposition}
\newtheorem{lemma}{Lemma}
\newtheorem{corollary}{Corollary}
\newtheorem{conjecture}{Conjecture}
\theoremstyle{remark}

\usepackage{microtype}
\usepackage{booktabs}
\usepackage{graphicx}
\graphicspath{{figures/}}
\usepackage{xr-hyper}
\usepackage{hyperref}
\usepackage{xurl}
\usepackage{xcolor}
\usepackage{orcidlink}

\newcommand{\thetafun}{\theta}
\newcommand{\Var}{\operatorname{Var}}

\newif\ifarxiv
\arxivtrue
\newif\iflead
\leadtrue

\ifarxiv\else\fi

\title{Random-walk scaling without Bernoulli statistics:\\
transport, persistence and resonances in the triangle map}
\author{Mirko Degli Esposti\,\orcidlink{0000-0003-0316-3449}\\[2pt]
\small Department of Physics and Astronomy, University of Bologna, Bologna, Italy\\
\small\href{mailto:mirko.degliesposti@unibo.it}{\texttt{mirko.degliesposti@unibo.it}}}
\date{}

\begin{document}
\maketitle

\begin{abstract}
The triangle map is an area-preserving, piecewise-parabolic map of the torus
with vanishing Lyapunov exponent. We study the symbolic process $X_n=\pm1$
given by the two halves of the torus, whose partial sums $S_n$ are the
transport of the lifted momentum, with an exact-phase integrator
and exact computations on polygonal cylinders. We prove
that the Koopman operator has no non-constant $H^1$ eigenfunctions, that the
symbol is not an $H^1$ coboundary, and that runs of equal symbols are
uniformly bounded. For generic parameters the walk is normally diffusive and
Gaussian, but its diffusion coefficient is not that of a coin: away from the
edges of the parameter square it varies between $0.2$ and $14$ times the
Bernoulli value, and for a benchmark pair, with $D=3.09\pm0.01$, the excess comes
from the order of the bounded inter-event intervals, not their distribution. The first-passage exponent is $1/2$ with a non-Bernoulli
amplitude, and worst-case Birkhoff averages decay as for a uniquely ergodic
system. For $\beta=0$ the Green--Kubo sum cancels, $D=0$, yet $\Var(S_n)$
keeps growing slowly, sign changes are rare and the correlations of Fourier
modes decay as $n^{-0.2}$. We identify the zero-sum words, which become pure
twists at $\beta=0$, as resonances that trap orbits with a computable
$1/\ell$ law, set the crossover time $\beta^{-1/2}$ at small $\beta$, and, in
an asymmetric deformation of the symbol, produce a logarithmic term in the
variance whose coefficient we compute exactly. We formulate conjectures on
unique ergodicity, mixing, the central limit theorem and the cohomology of the
symbol.
\end{abstract}

\iflead
\medskip
\begin{quote}\small
\noindent\textbf{Lead paragraph.}
Diffusion in deterministic systems is usually traced to chaos:
nearby orbits separate exponentially, the motion quickly forgets its past,
and the dynamics behaves like a sequence of coin tosses. Many systems,
however, diffuse without any exponential instability; billiards in polygons,
where nearby trajectories separate only linearly in time, are the classical
example. We study the triangle map, a simple model of the irrational
triangular billiard, and ask what its diffusion actually reveals. Scaling
alone turns out to say little: the symbolic walk spreads as the square root
of time with Gaussian fluctuations, as a coin would, but its diffusion
coefficient depends on the parameters and is fixed by the order in which
waiting times occur, not by their statistics. For a special parameter the
diffusion nearly stops. We trace this to resonant symbol patterns that trap
orbits for long times, and show how a small perturbation turns these traps
into measurable transport.
\end{quote}
\fi

\section{Introduction}

Deterministic dynamical systems may display statistical properties usually
associated with stochastic processes even in the absence of exponential
instability and positive
entropy~\cite{CasatiProsen1999,CasatiTsallisBaldovin2005,ArtusoOliveiraManchein2023}.
Polygonal billiards and related piecewise isometric systems are the natural
setting for this phenomenon~\cite{MasurTabachnikov2002}. Billiards in a
dense $G_\delta$ set of polygons are ergodic~\cite{KerckhoffMasurSmillie1986}
and weak mixing is also generic~\cite{ChaikaForni2026}; recent preprints
announce ergodicity and weak mixing for every triangle with an angle
irrational relative to $\pi$~\cite{OpenAIErgodicity2026,OpenAIWeakMixing2026},
building on transverse regularity and the cohomological equation for flows
on flat surfaces~\cite{ForniMoll2025}. Numerical studies have long reported
mixing-like, diffusive and anomalous transport in such
systems~\cite{CasatiProsen1999,HorvatEtAl2009,YoshidaEtAl2022,ArtusoOliveiraManchein2023},
whose dynamical origin is only partly understood.

The triangle map of Casati and Prosen~\cite{CasatiProsen2000}, a model of
the irrational triangular billiard, is one of the simplest systems in this
class. It is the area-preserving map of the torus
\begin{equation}
 \begin{gathered}
 p_{n+1}=(p_n+\alpha\,\thetafun(q_n)+\beta)\pmod 1,\\
 q_{n+1}=q_n+p_{n+1}\pmod 1,
 \end{gathered}
 \label{eq:triangle-map}
\end{equation}
where $\thetafun(q)=+1$ for $q\in[0,1/2)$ and $-1$ for $q\in[1/2,1)$.
In a previous study~\cite{HorvatEtAl2009} we found evidence of ergodic and
mixing behavior in a generic, non-resonant case, where the symbolic sums
\begin{equation}
 S_n=\sum_{j=0}^{n-1}X_j,
 \qquad
 X_j=\thetafun(q_j),
 \label{eq:symbolic-walk}
\end{equation}
scale as $n^{1/2}$, and much slower growth, slow correlation decay and
strong arithmetic sensitivity for $\beta=0$. That work introduced the
polygonal and symbolic descriptions of the map, the cubic growth of the
number of cylinders and finite Markov approximations.

Here we ask to what extent the stochastic-like behavior of the triangle map
is that of an uncorrelated random process. The $n^{1/2}$ scaling cannot
settle this: normal diffusion is compatible with strong temporal
correlations, which enter the diffusion coefficient through the Green--Kubo
relation~\cite{KlagesKorabel2002}, and decay of correlations does not imply
a positive diffusion coefficient, coboundaries being the standard
counterexample~\cite{Leonov1961}. We compare the map with an unbiased
Bernoulli sequence through three groups of diagnostics: the two-point
quantities that control transport ($\Var(S_n)$, Green--Kubo sums, the
spectrum of $X$); persistence, via first passages and the Sparre Andersen
theorem~\cite{SparreAndersen1953}; and the separation of the distribution of
inter-event intervals from their order~\cite{AltmannCristadoroDegliEsposti2012}.
Since double-precision orbits turn out to be unreliable for $\beta=0$, all
results use an integrator based on the exact lift of the dynamics, and many
quantities are computed without sampling from the polygonal cylinder sets.

The results separate three situations. The Bernoulli sequence has $D=1$. The
generic triangle map is normally diffusive, with Gaussian fluctuations and
first-passage exponent $1/2$, but its diffusion coefficient is not that of a
coin: $D=3.09\pm0.01$ for the benchmark pair $(\alpha,\beta)=(1/e,(\sqrt5-1)/2)$,
three times the Bernoulli value,
and $D$ varies between $0.2$ and $14$ away from the edges of the parameter
square, with deep minima at rational $\beta$. For the benchmark, $D$ is set by
the order of the inter-event intervals, which multiplies the value implied by
their distribution by $3.7$ ($2.4$--$3.7$ for other pairs), and the
persistence amplitude is not Bernoulli.
For $\beta=0$ the Green--Kubo sum cancels and $D=0$, while $\Var(S_n)$ keeps
growing slowly, almost independently of $\alpha$, sign changes of $S_n$ are
rare and Fourier modes decorrelate as $n^{-0.2}$. A small $\beta$ restores
diffusion after a time of order $\beta^{-1/2}$.

Three analytical observations frame and partly explain the numerics: the
Koopman operator has no non-constant $H^1$ eigenfunctions
(Proposition~\ref{prop:H1}), so the symbol is not an $H^1$ coboundary; runs of
equal symbols are uniformly bounded (Lemma~\ref{lem:runs}); and at $\beta=0$
every zero-sum word is a pure twist on its cylinder, carrying segments of
periodic points that trap nearby orbits for times with a $1/\ell$ law and an
exactly computable coefficient. These resonances fix the crossover time
$\beta^{-1/2}$, and an asymmetric deformation of the symbol, which makes them
displace the walk, turns them into a logarithmic term of the variance with a
coefficient we compute exactly and test numerically. We do not prove ergodicity or mixing;
worst-case Birkhoff averages over adversarially optimized initial conditions
decay as $N^{-1/2}$, as for the uniquely ergodic skew shift, and we state
precisely what the data support in three conjectures
(Sec.~\ref{sec:discussion}).

Section~\ref{sec:map} sets up the symbolic walk and the exact lift.
Section~\ref{sec:rigorous} contains the rigorous results.
Sections~\ref{sec:generic} and~\ref{sec:beta0} treat the generic and
$\beta=0$ regimes, Sec.~\ref{sec:fourier} the correlations of Fourier modes,
Sec.~\ref{sec:asymmetric} the asymmetric deformation, and
Sec.~\ref{sec:discussion} the conjectures. Numerical methods are in the
Appendix; further data and checks are in the supplementary material.

\section{The map and the symbolic walk}
\label{sec:map}

The discontinuity set of the map is $\mathcal D=\{q=0\}\cup\{q=1/2\}$. One
iteration is the composition $T=S\circ R_\beta\circ K_\alpha$ of the
discontinuous kick $K_\alpha(q,p)=(q,p+\alpha\thetafun(q))$, the translation
$R_\beta(q,p)=(q,p+\beta)$ and the parabolic shear $S(q,p)=(q+p,p)$. Away from
$\mathcal D$ the tangent map is
\begin{equation}
 A=\begin{pmatrix}1&1\\0&1\end{pmatrix},\qquad
 A^n=\begin{pmatrix}1&n\\0&1\end{pmatrix},
 \label{eq:parabolic-matrix}
\end{equation}
so nearby orbits that follow the same symbolic branch separate at most
linearly and the Lyapunov exponent vanishes. When a cut separates two orbits,
their kicks differ by $2\alpha$, and the shear converts this momentum
difference into a growing separation in $q$ until further cuts split them
again (Fig.~\ref{fig:mechanism}). The stochastic-like behavior is generated
by this interplay of parabolic shear and symbolic splitting, not by
exponential instability; the recurrence of orbits near the cuts and the
resulting orbit complexity were quantified in Ref.~\cite{DegliEspostiGalatolo2005}.

\begin{figure}[t]
  \centering
  \includegraphics[width=\textwidth]{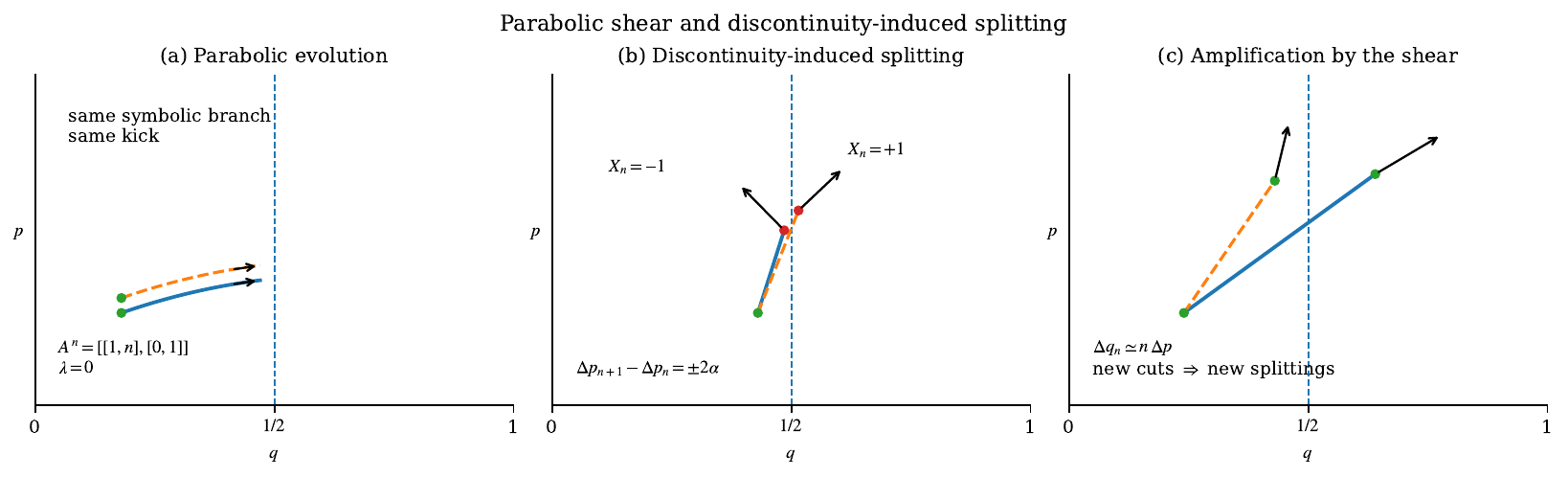}
  \caption{\textbf{Mechanism of the triangle map.}
  (a) Within a continuity domain nearby orbits separate parabolically.
  (b) A discontinuity gives them kicks differing by $2\alpha$.
  (c) The shear converts the momentum difference into separation in $q$,
  and further cuts generate new symbolic splittings.}
  \label{fig:mechanism}
\end{figure}

The symbolic sum $S_n$ is a transport observable. On the lift, iteration gives
\begin{equation}
 \begin{gathered}
 \widetilde p_n=p_0+n\beta+\alpha S_n,
 \qquad
 W_n=\sum_{k=1}^{n}S_k,\\
 \widetilde q_n=q_0+np_0+\beta\,\frac{n(n+1)}{2}+\alpha W_n,
 \end{gathered}
 \label{eq:exact-phase}
\end{equation}
with $(q_n,p_n)=(\widetilde q_n,\widetilde p_n)\bmod1$, so that
$\Var(\widetilde p_n-p_0-n\beta)=\alpha^2\Var(S_n)$. Equation~\eqref{eq:exact-phase}
also gives a non-accumulating integrator: $S_n$ and $W_n$ are integers, and
$q_n$ is recomputed at every step in double-double arithmetic, with an error
of order $10^{-16}$ that does not grow along the orbit (Appendix~\ref{app:methods}).
All results below use it.

Initial conditions are distributed according to Lebesgue measure, which is
invariant, so $X_n$ is stationary with $\langle X_n\rangle=0$. With
$C(k)=\langle X_0X_k\rangle$, $D_n=\Var(S_n)/n$ and
$G_K=C(0)+2\sum_{k=1}^{K}C(k)$, stationarity gives
\begin{equation}
 D_n=C(0)+2\sum_{k=1}^{n-1}\Bigl(1-\frac kn\Bigr)C(k)
 =\frac1n\sum_{K=0}^{n-1}G_K ,
 \label{eq:Dn-cesaro}
\end{equation}
so $D_n$ is the Ces\`aro mean of the Green--Kubo partial sums and differs
from $G_n$ at finite $n$. If $\sum_k|C(k)|<\infty$, both converge to the
diffusion coefficient
\begin{equation}
 D=\lim_{n\to\infty}D_n=G_\infty .
 \label{eq:gk-truncated}
\end{equation}
Note the normalization $\Var(S_n)\sim Dn$, without the factor $2$ common in
physics; the momentum diffusion coefficient in the convention
$\Var(\Delta\widetilde p_n)\sim2D_pn$ is $D_p=\alpha^2D/2$. The two-sided spectral density
$S_X(\omega)=\sum_kC(k)e^{-ik\omega}$ satisfies
\begin{equation}
 \Var(S_n)=\frac{1}{2\pi}\int_{-\pi}^{\pi}S_X(\omega)\,
 \frac{\sin^2(n\omega/2)}{\sin^2(\omega/2)}\,d\omega ,
 \label{eq:spectral-diffusion}
\end{equation}
so $S_X(0)=D$ when $S_X$ is continuous at zero. For the Bernoulli benchmark
$C(k)=0$ for $k\ge1$, $D=1$ and $S_X\equiv1$.

\section{Rigorous results}
\label{sec:rigorous}

This section collects the few statements we can prove. They concern the
spectrum of the Koopman operator, the cohomology of the symbol and the
combinatorics of the symbolic sequence; none of them establishes ergodicity,
weak mixing or mixing, but they frame the numerical results and, through the
word map below, partly explain them.

\subsection{No $H^1$ eigenfunctions and no $H^1$ transfer function}

The recent results on triangular billiards recalled in the
Introduction~\cite{OpenAIErgodicity2026,OpenAIWeakMixing2026} build on the
framework of Chaika and Forni and of Forni and
Moll~\cite{ChaikaForni2026,ForniMoll2025}, in which transverse regularity of
eigenfunctions plays a central role. This suggests the following elementary
discrete observation for our torus map. It is an analogy of method, not a
consequence of their statements about billiard flows, and it complements the
Fourier analysis of the Koopman operator in Ref.~\cite{HorvatEtAl2009}, where
eigenfunctions depending on $p$ alone were exhibited for rational $\alpha$.

\begin{proposition}\label{prop:H1}
If $\alpha$ is irrational and $\beta$ is arbitrary, every
$H^1(\mathbb T^2)$ eigenfunction of the Koopman operator
$Uf=f\circ T_{\alpha,\beta}$, for any eigenvalue $|\lambda|=1$ including
$\lambda=1$, is constant.
\end{proposition}

\begin{proof} Write $g=\partial_q f$, $h=\partial_p f$ and suppose
$Uf=\lambda f$. On the two continuity domains the derivative of
$T_{\alpha,\beta}$ is $A$, and the chain rule gives, almost everywhere,
$Ug=\lambda g$ and $U(g+h)=\lambda h$. Thus $Uh=\lambda(h-g)$ and
$U^n h=\lambda^n(h-ng)$. Since $U$ is unitary on $L^2$,
$\|h-ng\|_2=\|h\|_2$ for all $n$, so $g=0$ and $f(q,p)=F(p)$. The eigenvalue
equation on the two half-tori gives
$F(p+\beta+\alpha)=\lambda F(p)=F(p+\beta-\alpha)$, so $F$ is invariant under
rotation by $2\alpha$ and hence constant.
\end{proof}

Since $Uf=\lambda f\in H^1$, its global weak derivatives coincide with the
branchwise ones, and the jumps of $T$ do not contribute. The proposition does
not prove weak mixing, whose spectral criterion concerns all $L^2$
eigenfunctions: obtaining transverse regularity from a measurable
eigenfunction is precisely the hard part of the billiard results, whose
estimates do not control the jumps of our map.

Call $X=\thetafun$ a coboundary if $X=g-g\circ T$ almost everywhere for a
measurable $g$, so that $S_n=g-g\circ T^n$.

\begin{corollary}\label{cor:coboundary}
(i) If $X$ is a coboundary, then $-1$ is an eigenvalue of $U$, with
eigenfunction $e^{i\pi g}$; in particular $T^2$ is not ergodic.
(ii) For irrational $\alpha$ and any $\beta$, $X$ is not a coboundary with
$g\in H^1(\mathbb T^2)$.
\end{corollary}

\begin{proof} Since $e^{i\pi X}\equiv-1$, $f=e^{i\pi g}$ satisfies $Uf=-f$.
If $g\in H^1$, then $f\in H^1$, contradicting Proposition~\ref{prop:H1}.
\end{proof}

By a classical criterion for ergodic $T$~\cite[Theorem~11.8]{Schmidt1977},
$X$ is a measurable coboundary
if and only if the laws of $S_n$ are tight; for $L^2$ transfer functions,
if and only if $\Var(S_n)$ is bounded~\cite{Leonov1961}. Hence, if $T^2$ is
ergodic, $S_n$ cannot stay bounded in probability.

Since $X_j=\pm1$, $S_n\equiv n\pmod2$, so $(-1)^n=(-1)^{S_n}$ and
\begin{equation}
 (-1)^nX_n=(-1)^{S_n}\,(S_{n+1}-S_n).
 \label{eq:parity}
\end{equation}
The time-alternating symbol is therefore a function of the path of the walk
alone: a step from level $k$ contributes $(-1)^k$ times its direction. This
identity is used in Sec.~\ref{sec:resonances}.

\subsection{Bounded runs and the word map}
\label{sec:bounded-runs}

On the half-torus $\{\thetafun=s\}$, $s=\pm1$, the map is the skew shift
\begin{equation}
 F_c(q,p)=(q+p+c,\;p+c)\pmod 1,\qquad c=c_s=\beta+s\alpha .
 \label{eq:skew-shift}
\end{equation}

\begin{lemma}\label{lem:runs}
If $c_s$ is irrational, there is $L_s<\infty$ such that every orbit of
$T_{\alpha,\beta}$ contains no run of more than $L_s$ consecutive symbols
equal to $s$.
\end{lemma}

\begin{proof} For irrational $c$ the skew shift $F_c$ is
minimal~\cite{Furstenberg1961}: $F_c^n(q,p)=(q+np+c\,n(n+1)/2,\;p+nc)$, and for
$(k_1,k_2)\ne0$ the sequence $k_1q_n+k_2p_n$ is a polynomial in $n$ with
irrational leading coefficient, hence equidistributed by Weyl's criterion.
Let $U$ be the interior of $\{\thetafun=-s\}$. The open sets $F_c^{-k}U$,
$k\ge0$, cover $\mathbb T^2$, and by compactness finitely many, $k\le K$,
already do. A run of $r$ symbols $s$ has $F_{c_s}^i(x_j)\notin U$ for
$i<r$, so $r\le K$.
\end{proof}

The bound holds for every orbit, but it does not make every orbit dense. Even
density of all orbits would not, together with the invariance of Lebesgue
measure, imply ergodicity (minimal interval exchanges need not be uniquely
ergodic~\cite{KeynesNewton1976}), and unique ergodicity would not imply weak
mixing, as the skew shift itself ($\alpha=0$) shows. It excludes broad (infinite-variance) waiting-time
distributions, and it makes the interval statistics of
Sec.~\ref{sec:intervals} exactly computable. The bound is $L_s=L(c_s)$ with a
single function $L(c)=L(-c)=L(1-c)$, the longest run of $F_c$ in a
half-torus, so $L_+=L_-$ at $\beta=0$ but not in general (for the generic
benchmark $L_+=17$, $L_-=7$). $L(c)$ is small for most $c$ and diverges, as
$|c-a/b|^{-1/2}$, only at the rationals $a/b$ for which $F_{a/b}$ has a
periodic orbit inside one half-torus, such as $0$, $1/3$, $2/5$, $1/2$
(Fig.~\ref{fig:runs} in the supplementary material).

The argument extends to words. For $w=(w_0,\dots,w_{m-1})$ with partial sums
$s_k=w_0+\dots+w_{k-1}$, the lift gives on the cylinder $[w]$
\begin{equation}
 T^m(q,p)=\bigl(q+mp+d_w,\;p+c_w\bigr),\qquad
 c_w=m\beta+\alpha s_m,
 \label{eq:word-map}
\end{equation}
with $d_w=\beta m(m+1)/2+\alpha\sum_{k=1}^m s_k$. If $c_w$ is irrational,
repetitions of $w$ are uniformly bounded. At $\beta=0$, however, every
\emph{zero-sum} word ($s_m=0$) has $c_w=0$: on $[w]$ the block map is a pure
twist, and the lines $mp+d_w\in\mathbb Z$ carry sets of period-$m$ points.
These are the horizontal segments of periodic points of
Ref.~\cite{HorvatEtAl2009}; for $w=+-$ they lie at $p=1-\alpha/2$ and
$p=\tfrac12-\tfrac\alpha2$, with lengths $\alpha/2$ and $(1-\alpha)/2$. A
point at momentum distance $\delta$ from a segment of length $L$ is
translated along it by $m\delta$ per block, i.e.\ by $\delta$ per step, and
repeats $w$ until it leaves the segment, for a time $L/|\delta|$ if it enters
at one end (Fig.~\ref{fig:twist} in the supplementary material). The repetition containing a point,
i.e.\ the maximal stretch of consecutive steps along which its symbols repeat
$w$ periodically (about $\ell/m$ copies of $w$ for a stretch of $\ell$ steps),
is therefore at least $\ell$ steps long if $|\delta|\le L/\ell$, i.e.\ in a
strip of area $2L^2/\ell$ along the segment, and likewise along its $m-1$
images under $T,\dots,T^{m-1}$, which are translates of it. Since Lebesgue
measure is invariant, the probability that a given time belongs to a
repetition of $w$ of length at least $\ell$ (the fraction of time averaged over
Lebesgue-distributed initial conditions) is the area of these strips,
\begin{equation}
 F_w(\ge\ell)\simeq\frac{K_w}{\ell},\qquad K_w=2m\sum_jL_j^2,
 \label{eq:trap-law}
\end{equation}
the sum running over the segments of periodic points of $w$ itself, the
maximal intervals of $[w]$ on the resonant lines $p=(j-\alpha W)/m$,
$W=\sum_{k=1}^ms_k$. The images of these segments are the segments of the
cyclic rotations of $w$, so $K_w$ is the same for all rotations, and totals
over words count each cyclic class once. The law holds for $\ell\gg m$; at
finite $\ell$ the ends of the segments, where the strip meets the boundary of
$[w]$, give corrections (Sec.~\ref{sec:resonances}). For $w=+-$,
$K_{+-}=4[(\alpha/2)^2+((1-\alpha)/2)^2]=\alpha^2+(1-\alpha)^2$. For small
$\beta>0$ the drift $c_{+-}=2\beta$ makes $\delta$ grow linearly in time, so
that after $k$ steps the position along the segment is $\delta_0k+\beta k^2/2$;
the longest repetition is the one whose parabola just fits in the longer
segment, of length $L_{\max}=\max(\alpha,1-\alpha)/2$, which gives the maximal
length $2\sqrt{2L_{\max}/\beta}$. The lift also yields exact word
probabilities; for instance, at $\beta=0$, $P(-+-+)=1/12$ for every $\alpha$
(supplementary material, Sec.~\ref{sm:words}).

\subsection{Generic parameters and a correction}

We call $(\alpha,\beta)$ generic if $1$, $\alpha$, $\beta$ are linearly
independent over $\mathbb Q$. Then every $c_w$ is irrational, every word has
uniformly bounded repetitions, and the map has no periodic orbits, since a
period-$m$ orbit with word $w$ requires $c_w\in\mathbb Z$. All generic pairs
used below satisfy this condition, by the transcendence of $e$ and the linear
independence over $\mathbb Q$ of square roots of distinct square-free
integers, except $(\pi-3,1/e)$: the rational independence of $1$, $\pi$,
$1/e$ would follow from the algebraic independence of $\pi$ and $e$, which is
open.

\paragraph{Correction to Ref.~\cite{HorvatEtAl2009}.} There the generic case
was defined by $\alpha,\beta$ irrational with $\beta\notin\mathbb Z+\alpha\mathbb Z$,
and Lemma~1 states that the map then has no periodic orbits. The necessary
condition $t\beta+\alpha S_t\equiv0\pmod1$ for a period-$t$ orbit excludes
$\beta\in\mathbb Z+\alpha\mathbb Z$ only for $t=1$; in general it gives
$\beta\in\mathbb Q+\mathbb Q\alpha$, so the hypothesis must be strengthened
to rational independence. For example, $\alpha=1/e$, $\beta=(1-\alpha)/2$
satisfy the original condition, yet $w=+{+}{+}-$ has $c_w=4\beta+2\alpha=2$,
and the line $p\simeq0.97409$ carries a segment of length $0.342$ of
period-four points, which we verified in $50$-digit arithmetic. The other
rigorous statements of Ref.~\cite{HorvatEtAl2009} used here are not
affected.

\paragraph{Why ergodicity is hard.} On each half-torus the map is a uniquely
ergodic skew shift~\cite{Furstenberg1961}, and Veech proved strict ergodicity
of $\mathbb Z_2$-extensions of rotations by step cocycles~\cite{Veech1969};
both arguments need a factor. Here the lifted momentum is a rotation only
modulo the dense subgroup $\mathbb Z+\alpha\mathbb Z$: the symbol feeds back
into the base, and there is no factor. The billiard
theorems~\cite{KerckhoffMasurSmillie1986,ChaikaForni2026,OpenAIErgodicity2026,OpenAIWeakMixing2026}
rest on Teichm\"uller dynamics and do not transfer to a model map.
Boshernitzan-type criteria~\cite{Boshernitzan1992} need word complexity
growing linearly along a subsequence, whereas ours grows as $\simeq0.66\,n^3$
(supplementary material, Sec.~\ref{sm:words}), and the renormalization schemes for kicked
oscillators~\cite{KouptsovLowensteinVivaldi2002}, of which the triangle map
is the parabolic member $q_{n+1}-2q_n+q_{n-1}=\beta+\alpha\thetafun(q_n)$, do
not cover generic parameters. To our knowledge, minimality, unique
ergodicity and ergodicity of the generic map remain open; we state the
expected behavior as Conjecture~\ref{conj:generic} in
Sec.~\ref{sec:discussion}.

\section{Generic regime: correlated normal diffusion}
\label{sec:generic}

We use the benchmark $\alpha=e^{-1}$, $\beta=(\sqrt5-1)/2$ of our earlier
work.

\subsection{Diffusion coefficient}

For $4096$ orbits followed to $n=10^8$, the local log-log slopes of
$\Var(S_n)$ and of the mean running maximum
$\langle S_{\max}(n)\rangle=\langle\max_{k\le n}S_k\rangle$ are $0.996$ and
$0.502$, and $\Var(S_n)/n=3.079$ at the largest checkpoint; the amplitude
$\langle S_{\max}(n)\rangle/\sqrt n=1.408$ matches the Brownian value
$\sqrt{2D/\pi}=1.400$. Correlations from $64$ trajectories of length $2^{26}$
agree with the values computed without sampling from the polygonal cylinders
(App.~\ref{app:methods}) for $k<80$: the mean squared deviation in units of
the bootstrap error is $0.88$ per lag, as expected for purely statistical
scatter.
One has $C(1)=C(2)=0$ exactly for all parameters (for fixed $q_0$, the
uniform $p_0$ makes $q_1$ uniform, and
$q_2=2q_1-q_0+\alpha\thetafun(q_1)+\beta$ is the image of $q_1$ under a
two-to-one map that preserves Lebesgue measure), but the next correlations are large, $C(4)\simeq0.081$,
$C(7)\simeq0.167$ (supplementary material, Fig.~\ref{fig:correlations}), and $G_K$ reaches $2.04$
at $K=10$ and $2.66$ at $K=30$ (Fig.~\ref{fig:transport-regimes}, middle
left).
The rest of the excess is built by a band of weaker, sign-oscillating
correlations at $10^2\lesssim k\lesssim2\times10^3$, after which $G_K$ is
flat: $G_{10^3}=3.064\pm0.003$, $G_{10^4}=3.091\pm0.009$. By
Eq.~\eqref{eq:Dn-cesaro}, $D_n=D-2M/n+o(1/n)$ with $M=\sum_kkC(k)\simeq100$,
which explains the slow approach of the finite-time variance. A
non-oscillating tail $C(k)\simeq ak^{-3/2}$, as reported for the smooth
observable $\cos\pi y$~\cite{CasatiProsen2000}, is not seen ($a=0.08\pm0.16$ over
$2\times10^3\le k\le10^4$). The averaged periodogram is flat below
$\omega\sim10^{-3}$ at $3.11\pm0.05$ (bottom left). All estimators agree
(Fig.~\ref{fig:transport-regimes}, left):
\begin{equation}
 \Var(S_n)\sim Dn,\qquad S_X(\omega)\to D,\qquad D=3.09\pm0.01 .
 \label{eq:generic-summary}
\end{equation}
The value is not universal: away from the edges of the parameter square $D$
ranges from $0.2$ to $14$ (median $2.3$), with dark lenses at rational $\beta$ and
bright bands near low-order resonances (supplementary material, Fig.~\ref{fig:Dmap}); four other
generic pairs give $D=2.15$--$2.94$ (supplementary material, Sec.~\ref{sm:generic}), and the pair A
of Sec.~\ref{sec:asymmetric} gives $D=1.19$, close to the Bernoulli value. The
generic symbolic walk is \emph{normally diffusive but correlated}.

The fluctuations are Gaussian on the explored scales: the Kolmogorov distance
from the standard Gaussian of the normalized sum $S_n/\sqrt{\Var(S_n)}$ is $0.013$--$0.017$ for
$10^3\le n\le3.4\times10^7$, at the level of sampling noise, the skewness
decays as $n^{-1/2}$, and the correlation between $S_n$ and $S_{2n}-S_n$ is
below $0.02$ for $n\ge10^4$. This supports, without proving, the central limit
theorem of Conjecture~\ref{conj:generic} (Sec.~\ref{sec:discussion}), with
asymptotically independent increments~\cite{DuarteQueiros2008}.

\subsection{Distribution versus order of inter-event intervals}
\label{sec:intervals}

Let $\tau_i$ be the intervals between consecutive occurrences of $+1$,
observed from an occurrence at time $0$, and let
$t_J=\tau_0+\dots+\tau_{J-1}$ be the time of the $J$-th return. By Kac's
lemma the mean interval is $2$, and since $[0,t_J)$ contains $J$ symbols $+1$
and $t_J-J$ symbols $-1$, $\sum_{i<J}(\tau_i-2)=-S_{t_J}$ exactly. Hence $D=\tfrac12G(\infty)$ with
$G(K)=v+2\sum_{k=1}^{K}\operatorname{Cov}(\tau_0,\tau_k)$, $v=\Var(\tau_0)$,
whereas an independent renewal process with the same interval law has
$D_R=v/2$~\cite{AltmannCristadoroDegliEsposti2012}. By Lemma~\ref{lem:runs}
the intervals are bounded, and the set of initial conditions producing a
given run is a finite union of polygons, so the interval law is known
exactly and no shuffled surrogate is needed:
\begin{center}
\begin{tabular}{lcccc}
\toprule
 & longest runs $+/-$ & $\max\tau$ & $v$ & $D_R=v/2$\\
\midrule
generic (benchmark pair) & $17/7$ & $8$ & $1.7085$ & $0.854$\\
$\beta=0$, $\alpha=1/e$ & $11/11$ & $12$ & $2.1493$ & $1.075$\\
$\beta=0$, $\alpha=\sqrt2-1$ & $12/12$ & $13$ & $2.1254$ & $1.063$\\
Bernoulli & $\infty$ & $\infty$ & $2$ & $1$\\
\bottomrule
\end{tabular}
\end{center}
If the intervals were independent with the same law, the walk would have
$D_R=v/2$: $0.85$ in the generic regime and $1.07$ at $\beta=0$, both close to
the Bernoulli value. The measured transport is very different in both
regimes, so it is set by the order of the intervals, i.e.\ by the covariances
in $G(K)$ [Fig.~\ref{fig:persistence}(c)]. For the benchmark pair they are
positive (long intervals tend to follow long ones) and $\tfrac12G(K)$ grows
from $0.85$ to $3.15\pm0.03$ at $K=4096$, a factor of about $3.7$, consistent
with the more precise $D=3.09\pm0.01$ obtained from the symbol sum; $80\%$
of the increase occurs within twenty events. At $\beta=0$ they alternate in sign,
starting with $\operatorname{Cov}(\tau_0,\tau_1)\simeq-0.43$ (a long interval
is followed by a short one that compensates it), and $\tfrac12G(K)$ falls
from $1.07$ to $0.05$--$0.15$ within about ten events; the remainder,
$\tfrac12G\simeq0.005$ at $K\sim10^3$, decays slowly. Bounded intervals and
their order thus produce both enhanced and suppressed transport, whereas
anomalous transport is usually attributed to broad waiting-time
distributions.

\subsection{Persistence}
\label{sec:persistence}

Let $Q_\pm(n)$ be the probability that the walk has not entered the positive
(negative) half-line up to time $n$, i.e.\ that $\pm S_k\le0$ for all
$1\le k\le n$. For a Bernoulli walk (independent symmetric steps $\pm1$, hence
the subscript) the two coincide and the Sparre Andersen
theory~\cite{SparreAndersen1953,SparreAndersen1954} gives
\begin{equation}
 Q_B(n)=\binom{n}{\lfloor n/2\rfloor}2^{-n}\sim\sqrt{\frac{2}{\pi n}} ,
 \label{eq:sa-lattice}
\end{equation}
the lattice counterpart of the universal $n^{-1/2}$ law for continuous
symmetric steps.
The involution $(q,p)\mapsto(-q,-p)$ preserves Lebesgue measure, reverses
$\thetafun$ and conjugates $T_{\alpha,\beta}$ to $T_{\alpha,-\beta}$, so
$Q_+(n;\alpha,\beta)=Q_-(n;\alpha,-\beta)$ exactly. For the generic parameters
($10^6$ orbits to $n=2^{22}$) the effective exponent
$\theta_{\rm eff}=-d\ln Q/d\ln n$ is compatible with $1/2$ for $n\gtrsim10^2$,
but the amplitudes are not Bernoulli and differ between the sides,
$Q_+/Q_B\simeq1.75$ and $Q_-/Q_B\simeq1.32$
[Fig.~\ref{fig:persistence}(a,b)]. They are not fixed by the one-time
marginals either. For independent steps the survival probability is
determined by them through the Sparre Andersen identity
$\sum_{n\ge0}Q_+(n)s^n=\exp\bigl(\sum_{k\ge1}s^kP(S_k\le0)/k\bigr)$~\cite{SparreAndersen1954};
fed with the measured $P(S_k\le0)$ (and $P(S_k\ge0)$ for $Q_-$), it
reproduces simulated Bernoulli walks within $1\%$, but for the map the
measured $Q_+$ and $Q_-$ exceed its prediction by factors $1.38$ and $1.77$ at
$n=4096$: the excess persistence comes from temporal correlations, not from
the distribution of $S_k$.

At $\beta=0$ persistence is anomalously strong: with $10^6$ orbits to
$n=1.7\times10^7$, $Q/Q_B$ reaches $28$, and $\theta_{\rm eff}$ decreases from
$0.29$ at $n\sim10^2$ to $0.205\pm0.014$ for $10^6\lesssim n\le1.7\times10^7$,
identically for $\alpha=1/e$ and $\sqrt2-1$; a logarithmic decay is excluded,
and no limiting exponent is claimed. The long survivors are selected early by
thin bands of initial momentum (supplementary material, Sec.~\ref{sm:beta0}).

\begin{figure}[t]
  \centering
  \includegraphics[width=\textwidth]{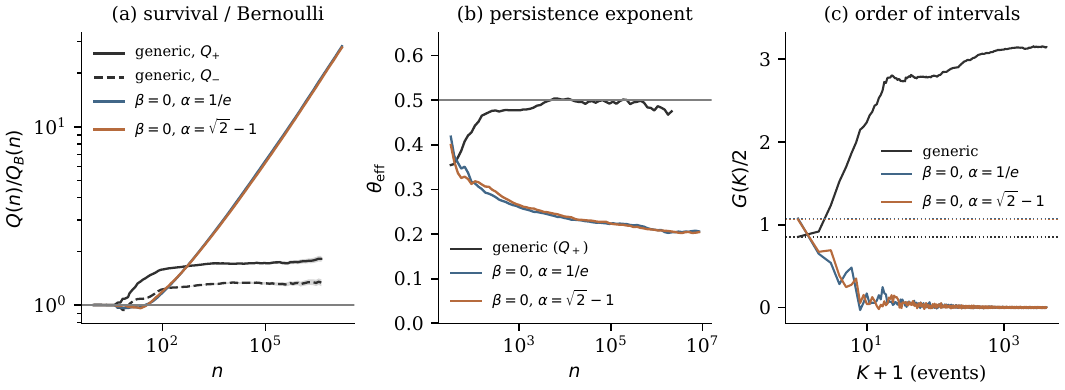}
  \caption{\textbf{Persistence and inter-event intervals.}
  (a) Survival probabilities relative to the lattice Bernoulli reference
  $Q_B$, Eq.~\eqref{eq:sa-lattice}: generic $Q_+$ (solid) and $Q_-$ (dashed)
  with Wilson $68\%$ bands, and $(Q_++Q_-)/2$ for $\beta=0$.
  (b) Effective persistence exponent over factor-four windows; the
  horizontal line is the Sparre Andersen value $1/2$.
  (c) Event-time Green--Kubo sums $G(K)/2$ of the intervals between $+$
  events; dotted lines mark the renewal values $v/2$ computed exactly from
  the interval law.}
  \label{fig:persistence}
\end{figure}

\section{The exceptional $\beta=0$ regime}
\label{sec:beta0}

\subsection{Variance, correlations and spectrum}

At $\beta=0$ double-precision orbits are not adequate: their symbol sequences
differ from the exact ones from $n\sim10^6$, and for an $\alpha$ with large
partial quotients the double-precision variance is three times too large at
$n\gtrsim5\times10^6$ (supplementary material, Sec.~\ref{sm:beta0}). We followed $2048$ orbits to
$n=2^{25}\simeq3.4\times10^7$ for $\alpha\in\{\sqrt2-1,\sqrt3-1,(\sqrt5-1)/2,1/e\}$
and for two values with large partial quotients, $\pi-3$ and
$[0;2,30,2,300,2,3000,1,1,\dots]$ (continued-fraction expansion), with replicate ensembles and one run to
$1.3\times10^8$. The variance curves are nearly independent of $\alpha$:
$\Var(S_n)\simeq13$, $33$--$38$, $65$--$73$ and $95$--$103$ at $n=10^2$,
$10^4$, $10^6$ and $3.4\times10^7$ (Fig.~\ref{fig:transport-regimes}, right),
with no feature at the convergent denominators. The local exponent
$d\ln\Var/d\ln n$ decreases from $0.3$ to $0.1$ without a plateau, and
$d\Var/d\ln n$ increases, so the growth is slower than any fixed accessible
power and faster than $\ln n$. Beyond this the law is not determined:
$a+b\ln^2n$, $a+c\ln n+b\ln^2n$, $a+b(\ln n)^p$ and $a+cn^\gamma$ fit equally
well, with model-dependent parameters (supplementary material, Sec.~\ref{sm:beta0}).

The Green--Kubo sums explain the suppression. From $1024$ trajectories of
length $2^{22}$, $G_K$ falls from $1$ to $0.45$ at $K=10$ and $0.09$ at
$K=100$, then oscillates around zero with decreasing amplitude ($-0.029$,
$0.022$, $-0.014$ at $K=10^3,10^4,10^5$, errors $0.002$--$0.004$): the
positive instantaneous variance is almost exactly cancelled by integrated
anticorrelations, $\sum_{k\ge1}C(k)\simeq-1/2$. The spectral density vanishes
at low frequency with local exponent $\eta(\omega)$ close to one and
arithmetic oscillations; scale by scale $1-\eta(2\pi/n)$ matches the local
variance exponent, as Eq.~\eqref{eq:spectral-diffusion} requires, and a
log-corrected linear spectrum $S_X\sim|\omega|\ln(1/|\omega|)$, which gives
$\ln^2n$, is consistent with both drifts without being singled out. Exact Markov
approximations of order $J$, built without sampling on the return map to the
starts of $+$ runs (supplementary material, Sec.~\ref{sm:markov}), show the same structure on
short scales ($n\lesssim10^2$): at $\beta=0$ their diffusion coefficient falls
as $D_J\simeq0.78/J$ ($14\le J\le28$), against $D_J\simeq2.9$ in the generic
case, and their transfer function, the solution of the Poisson equation,
grows only logarithmically in $J$ instead of as $\sqrt J$, so that at every
finite order $X$ is close to a coboundary, as in
Conjecture~\ref{conj:cohomology}.

Thus, on the explored scales ($n\le3.4\times10^7$, $K\le1.3\times10^5$,
$\omega\gtrsim3\times10^{-6}$), $\Var(S_n)/n\to0$, $D=0$ and $S_X(\omega)\to0$,
while $\Var(S_n)$ grows without saturation. If the growth persists, $X$ is
not an $L^2$ coboundary (Sec.~\ref{sec:rigorous}: an $L^2$ transfer function
would keep $\Var(S_n)$ bounded), although its Green--Kubo sum vanishes as for
a coboundary. The bulk of
the distribution widens (interquartile range $4\to12$ from $n=10$ to
$3.4\times10^7$), and $S_n/\sigma_n$ has a stable, $\alpha$-independent,
non-Gaussian law with kurtosis $\simeq5$.

\begin{figure}[t]
  \centering
  \includegraphics[width=\textwidth]{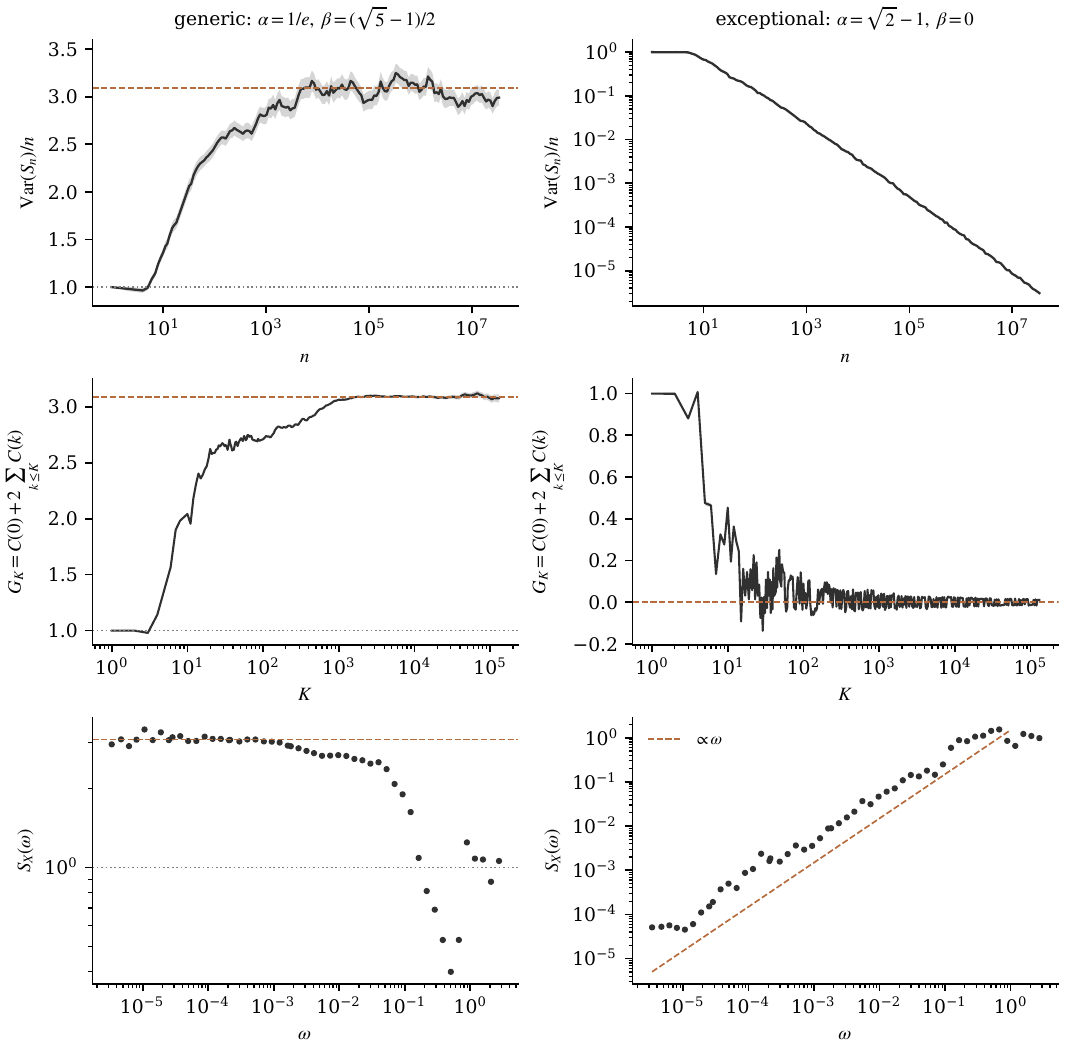}
  \caption{\textbf{Two transport regimes.}
  Left: generic $\alpha=1/e$, $\beta=(\sqrt5-1)/2$. Right: $\alpha=\sqrt{2}-1$,
  $\beta=0$. Top: $\Var(S_n)/n$ ($2048$ orbits). Middle: Green--Kubo partial
  sums $G_K$ ($64$ trajectories of length $2^{26}$, generic; $1024$ of
  length $2^{22}$, $\beta=0$). Bottom: log-binned spectral density
  ($32$ trajectories of length $2^{25}$). Bands are one standard error;
  dashed lines mark $D=3.09$, $0$ and $S_X\propto\omega$; dotted lines the
  Bernoulli value $1$. All data from the exact-phase integrator.}
  \label{fig:transport-regimes}
\end{figure}

\subsection{Parity and zero-sum-word resonances}
\label{sec:resonances}

The walk $S_n$ moves on the integer lattice, and each step crosses one edge
$[k,k+1]$ of it ($k\in\mathbb Z$ a level, not a symbol). By the parity
identity~\eqref{eq:parity}, the alternating sum $A_n=\sum_{j<n}(-1)^jX_j$ adds
$(-1)^k$ for every crossing of the edge $[k,k+1]$, in either direction: $+1$
for edges with an even lower end and $-1$ for edges with an odd one. A walk that oscillates over few edges, as in a repetition of
$+-$, therefore has an almost ballistic $A_n$, whereas a walk spreading over
many levels crosses both kinds equally often. Since $A_n/n=\frac1n\sum_{j<n}(-U)^jX$
is an ergodic average for the unitary operator $-U$, the mean ergodic theorem
gives $A_n/n\to P_{-1}X$ in $L^2$, the projection of $X$ on the eigenspace
$\{Uf=-f\}$; equivalently, $\mathbb E|A_n/n|^2\to\mu_X(\{\pi\})$, the mass at
the frequency $\pi$, that of the period-two oscillation $(-1)^j=e^{i\pi j}$, of
the spectral measure of $X$, defined by $C(k)=\int e^{ik\omega}d\mu_X(\omega)$.
A positive limit would give an eigenfunction with $Uf=-f$, hence a
non-constant $T^2$-invariant function: $T^2$ would not be ergodic, the same
obstruction that a coboundary would create (Corollary~\ref{cor:coboundary}).
Conjecture~\ref{conj:cohomology} asserts that neither occurs. In the generic case $\Var(A_n)/n\simeq1.7$ is constant. At
$\beta=0$, $\Var(A_n)/n^2$ decreases slowly, from $0.051$ to $0.0091$
($\alpha=\sqrt2-1$) between $n=2^{10}$ and $2^{22}$, as $n^{-0.2}$, and $S_X$
has a peak at $\omega=\pi$ of height $\simeq900$ against a median of $0.8$,
with weaker peaks at $2\pi/3$, $\pi/2$, $\pi/3$, $\pi/4$. At $\beta=0$ the
momentum stays on the lattice $p_0+\alpha\mathbb Z$, and by
Eq.~\eqref{eq:parity} the sign of each $+-$ burst is fixed by the level on
which it occurs, so successive bursts add coherently at $\omega=\pi$.

These resonances are seen directly in the orbits. Let $F_m(\ge\ell)$ be the
fraction of time spent inside repetitions, at least $\ell$ steps long, of
zero-sum words of length $m$; Eq.~\eqref{eq:trap-law} predicts
$\ell F_m(\ge\ell)=K_m=\sum_wK_w$, independent of $\ell$. At $\beta=0$ each
orbit keeps its momentum lattice $p_0+\alpha\mathbb Z$ and samples a single
distance $\delta$ from each segment, so we use $\beta=10^{-7}$, which lets
the momentum drift while cutting repetitions only beyond
$2\sqrt{2L_{\max}/\beta}\simeq5\times10^3$. Over $10m\le\ell\le100m$, for
$\alpha=\sqrt2-1$ and $m=2,4,6,8$, we find $\ell F_m=0.52$--$0.59$,
$1.48$--$1.55$, $0.38$--$0.43$, $0.53$--$0.58$ against the exact
$K_m=0.515$, $1.373$, $0.373$, $0.547$; $\ell F_m$ decreases with $\ell$
towards $K_m$, the excess at $\ell\simeq10m$ being a finite-$\ell$ correction
[the same law for the deformed map is shown in Fig.~\ref{fig:asymmetric}(b)]. For $\beta=10^{-3},\dots,10^{-6}$ the
longest $+-$ repetition agrees with the cutoff $2\sqrt{2L_{\max}/\beta}$ within
$1$--$5\%$ [supplementary material, Fig.~\ref{fig:twist}(c)]. The $K_m$ do not decrease with $m$:
resonant twists exist at all lengths.

\subsection{Crossover at small $\beta$}

For $\alpha=1/e$ and $\sqrt2-1$ and seven values $6\times10^{-5}\le\beta\le0.06$
($1024$ orbits to $n=2^{25}$ each), the variance follows the $\beta=0$ curve
up to a crossover time and then grows linearly. The ratio
$\Var_\beta(S_n)/\Var_0(S_n)$ is approximately a function of $n\sqrt\beta$
(scatter $\pm30\%$, not monotone in $\beta$, larger than the statistical
errors), with crossover time $n_c\simeq(2$--$3)\,\beta^{-1/2}$. The relevant
time is not $1/\beta$, at which the drift $n\beta$ is of order one, but the
time at which the quadratic phase $\beta n(n+1)/2$ of Eq.~\eqref{eq:exact-phase}
is. It is the lifetime of the resonances: for $\beta>0$ the block map
of a zero-sum word acquires the drift $c_w=m\beta$, and the longest $+-$ repetition is
$2\sqrt{2L_{\max}/\beta}\simeq1.5\,\beta^{-1/2}$. At the smallest $\beta$,
\begin{equation}
 D(\beta)\simeq\sqrt\beta\;\Var_0\bigl(S_{\beta^{-1/2}}\bigr),
 \label{eq:crossover}
\end{equation}
with ratio of the two sides $1.00$--$1.07$ ($\alpha=1/e$) and $0.8$--$1.25$
($\sqrt2-1$) for $\beta\lesssim6\times10^{-4}$: the variance accumulated by
the $\beta=0$ dynamics up to the crossover is then spread diffusively.

\section{Correlations of Fourier modes and mixing}
\label{sec:fourier}

Mixing is equivalent to the decay of $\langle f,U^nf\rangle$ for all
zero-mean $f\in L^2$, and it suffices to check the Fourier modes
$f_{k,l}=e^{2\pi i(kq+lp)}$. Their correlations were measured in double
precision in Refs.~\cite{CasatiProsen2000,HorvatEtAl2009}, with decay close
to $n^{-3/2}$ (generic) and $n^{-0.2}$ ($\beta=0$). We repeated the
measurement with the exact lift, recomputing $q_n$ and
$p_n=p_0+n\beta+\alpha S_n\bmod1$ at every step, for the modes $(1,0)$,
$(0,1)$, $(1,1)$, $(2,0)$, $(0,2)$, with $3$--$4\times10^9$ steps per
parameter pair (the mode $(1,-1)$ coincides with $(1,0)$ because
$q_n-p_n=q_{n-1}$).

In the generic case the correlations decay quickly
(Fig.~\ref{fig:fourier}): the rms $|C_{1,0}|$ is $7\times10^{-3}$ for
$30\le n<100$ and $5\times10^{-4}$ for $300\le n<1000$, with effective
exponents $1.4$--$2$, and the signal is below the noise floor
$6\times10^{-5}$ beyond $n\simeq10^3$; this is consistent with $n^{-3/2}$
over about one decade. At $\beta=0$ all modes decay slowly but steadily up to
$n=2.6\times10^5$, for $\alpha=1/e$ and $\sqrt2-1$, with effective exponent
$0.19$--$0.21$ for $(1,0)$, $(1,1)$, $(0,1)$, and no plateau. The
correlations are real: the symmetry $(q,p)\mapsto(-q,-p)$ conjugates the map
with $\beta$ to the map with $-\beta$, so at $\beta=0$ every $C_{k,l}$ is real.
(For any $\beta$ the modes $(k,0)$ are real as well, because the map is
reversible with the $q$-preserving reversor $(q,p)\mapsto(q,-p-\alpha\theta(q)-\beta)$;
the other modes are complex for $\beta\ne0$.) They carry the parity of the walk: the $q$-modes nearly vanish at odd lags, and
the $p$-modes alternate in sign [Fig.~\ref{fig:fourier}(c)].

The $p$-modes have a direct meaning for the cocycle: at $\beta=0$,
\begin{equation}
 C_{0,l}(n)=\mathbb E\,e^{2\pi il\alpha(S_{j+n}-S_j)}
 \label{eq:char-function}
\end{equation}
is the characteristic function of the increments of $S$. If $X$ were a
coboundary and $T$ mixing, it would tend to $|\mathbb E\,e^{2\pi il\alpha g}|^2$,
generically non-zero. The steady decay of $C_{0,1}$, from $0.39$ at
$n\simeq10$ to $0.094$ at $n\simeq5\times10^4$, is evidence against
tightness. Finally, worst-case Birkhoff averages of the Fourier modes over
random, near-cut and adversarially optimized initial conditions decay as
$N^{-1/2}$ for two generic pairs, like the uniquely ergodic skew shift, while
the same search recovers the periodic segments at $\beta=0$ and the
period-four orbits of the counterexample above (supplementary material, Sec.~\ref{sm:worst}).

\begin{figure}[t]
  \centering
  \includegraphics[width=\textwidth]{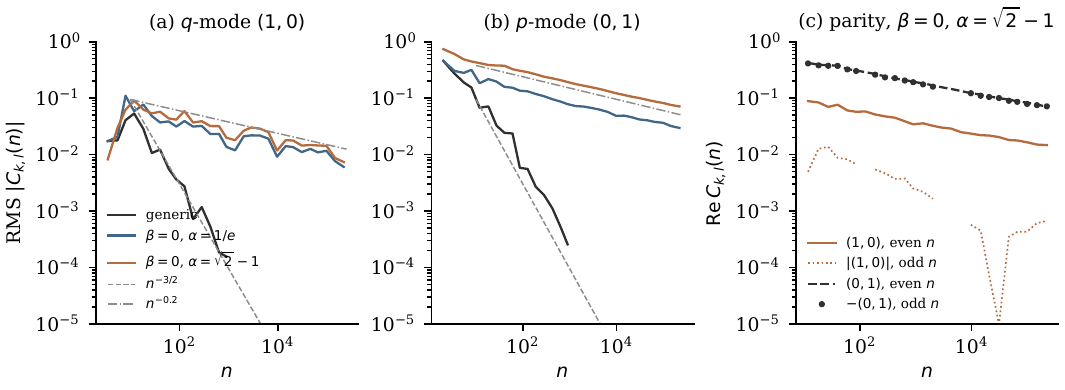}
  \caption{\textbf{Correlations of Fourier modes} (exact lift).
  (a), (b) Root-mean-square $|C_{k,l}(n)|$ over logarithmic lag bins, after
  subtraction of the statistical noise, for the modes $(1,0)$ and $(0,1)$,
  shown where the signal exceeds three standard errors (generic noise floor
  $\simeq6\times10^{-5}$). Grey lines: $n^{-3/2}$ (dashed) and $n^{-0.2}$
  (dash-dotted). (c) Parity structure at $\beta=0$, $\alpha=\sqrt2-1$: real
  part of $C_{1,0}$ at even lags and its modulus at odd lags, and $C_{0,1}$ at
  even lags compared with $-C_{0,1}$ at odd lags.}
  \label{fig:fourier}
\end{figure}

\section{An asymmetric deformation: resonances made visible}
\label{sec:asymmetric}

At $\beta=0$ the triangle map has two features that cannot be separated: the
zero-sum-word resonances and the zero mean drift of the lifted momentum.
Replacing $\thetafun$ by
\[
  \thetafun_\varepsilon(q)=\begin{cases}-1,& q\in[0,\varepsilon),\\ +1,& q\in[\varepsilon,1),\end{cases}
\]
in the same map interpolates between the triangle map ($\varepsilon=1/2$, up
to $\alpha\to-\alpha$ and $q\to q+1/2$) and the skew shift ($\varepsilon=0$);
the lift~\eqref{eq:exact-phase} and the exact integrator are unchanged. Now
$\mathbb E\,S_n=n(1-2\varepsilon)$, and we consider the centered walk
$X_n=S_n-n(1-2\varepsilon)=-2(N_n-n\varepsilon)$, with $N_n$ the number of visits
to $[0,\varepsilon)$, and the Fano factor of these counts,
$F=\Var(N_n)/\mathbb E\,N_n=\Var(S_n)/(4\varepsilon n)$
(Poisson: $F=1$; Bernoulli: $F=1-\varepsilon$).

\paragraph{Generic $\beta$.} For small $\varepsilon$ the kicks are the
returns of the quadratic sequence $q_0+np_0+c_+n(n+1)/2$ to a short interval,
whose statistics are Poissonian for generic
coefficients~\cite{RudnickSarnak1998,Marklof2003}. We consider two generic
pairs, A $=(\sqrt2-1,\sqrt3-1)$ and B $=((\sqrt5-1)/2,\sqrt3-1)$. Numerically
$F$ stays within $0.96$--$1.12$ (A) and $0.96$--$1.22$ (B) for
$\varepsilon<0.04$ [Fig.~\ref{fig:fano}(a),(b)], and the $p$-modes follow the
Poisson prediction $\exp[-n\varepsilon(1-\cos4\pi l\alpha)]$ for
$\varepsilon\le0.03$ (not shown). At larger
$\varepsilon$ the Fano factor is strongly non-monotone and reproducible
between independent ensembles (Fig.~\ref{fig:fano}): for pair A it peaks at
$F=3.18$ at $\varepsilon=0.29$, and the diffusion coefficient
$D=4\varepsilon F$ peaks at $4.09$ at $\varepsilon=0.374$ [panel (c)]. The
triangle map with the same parameters ($\varepsilon=1/2$) is thus far from the
most diffusive member of the family: for pair A its $D=1.19$ (against $3.09$
for the benchmark pair of Sec.~\ref{sec:generic}) is close to the Bernoulli
value $1$ and $3.4$ times smaller than the maximum. At $\varepsilon=1/4$ the variance saturates,
$\Var(X_n)/n=2.62\pm0.01$ from $n\simeq2^{13}$ to $2^{24}$.

\begin{figure}[t]
  \centering
  \includegraphics[width=\textwidth]{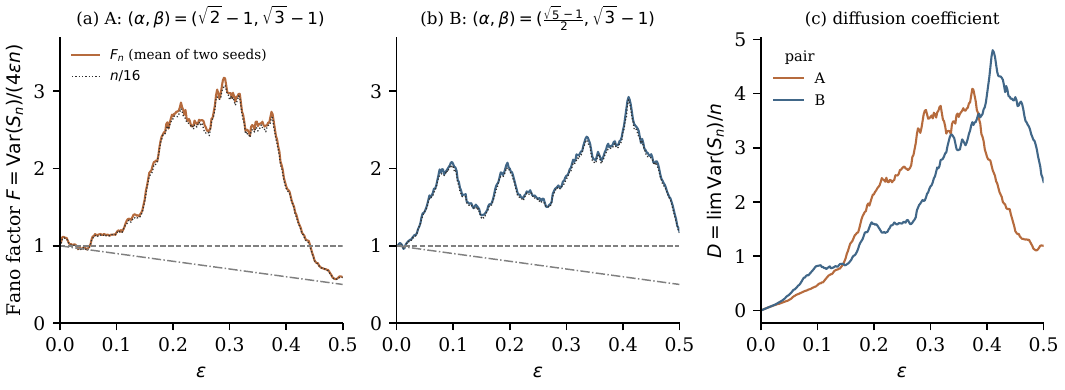}
  \caption{\textbf{Fano factor of the deformed map} at generic $\beta$,
  $250$ values of $\varepsilon$, $n=2^{16}$, two independent ensembles of
  $32\,768$ orbits per pair. (a), (b) $F=\Var(S_n)/(4\varepsilon n)$ for pairs
  A and B: mean of the two ensembles (band: their spread; dotted: $F$ at
  $n/16$); dashed, the Poisson value $1$; dash-dotted, the Bernoulli value
  $1-\varepsilon$. (c) $D=4\varepsilon F$.}
  \label{fig:fano}
\end{figure}

\paragraph{$\beta=0$: ballistic stays and a logarithmic term.} At $\beta=0$
and $\varepsilon\ne1/2$ the variance per step grows. At $\varepsilon=1/4$
($10^5$ orbits to $2^{24}$, two seeds) the raw variance is dominated by a
few orbits (one orbit carries up to $18\%$ of $\sum X_n^2$, Hill index of
$|X_n|$ about $2$), while the interquartile scale is reproducible and grows
linearly in $\ln n$ [Fig.~\ref{fig:asymmetric}(a)]. Each extreme orbit sits,
for up to $\gtrsim2.5\times10^6$ steps, on a periodic stretch of period two
($+-$) or four ($++--$), i.e.\ on a zero-sum-word resonance. During such a
stay $S_n$ does not move while the centring term does, so $X_n$ changes by
$-(1-2\varepsilon)$ per step: a ballistic excursion. At $\varepsilon=1/2$ the
same stays do not displace the walk.

The stays obey the law~\eqref{eq:trap-law}, with $K_w$ computed with the
threshold $\varepsilon$ and no free parameter [Fig.~\ref{fig:asymmetric}(b)]:
in $20\,000$ orbits of length $2^{24}$ ($\alpha=\sqrt2-1$),
\begin{center}
\begin{tabular}{@{}lccc@{}}
\toprule
 & $\varepsilon=0.1$ & $\varepsilon=0.2$ & $\varepsilon=0.25$\\
\midrule
$K_2$ (theory / measured) & $0.0800\ /\ 0.080$ & $0.3200\ /\ 0.318$--$0.320$ & $0.4216\ /\ 0.424$\\
$K_4$ (theory / measured) & $0.0800\ /\ 0.078$--$0.079$ & $0.3200\ /\ 0.322$ & $0.5000\ /\ 0.499$--$0.501$\\
\bottomrule
\end{tabular}
\end{center}
The number of stays of length at least $\ell$ per step is then
$K/(2\ell^2)$, and a stay of length $\ell$ adds $(1-2\varepsilon)^2\ell^2$ to
$X^2$: the second moment is marginal (hence the Hill index $\simeq2$ above),
and
\begin{equation}
  \frac{\Var(X_n)}{n}\simeq D_0+c(\varepsilon)\ln n,\qquad
  c(\varepsilon)=(1-2\varepsilon)^2K_{\rm tot},\quad K_{\rm tot}=\sum_wK_w .
  \label{eq:clog}
\end{equation}
For $\alpha=\sqrt2-1$ and $\varepsilon\le0.2$ the enumeration finds only
$m=2,4$, with
$K_2=K_4=8\varepsilon^2$, so $c(\varepsilon)=16\varepsilon^2(1-2\varepsilon)^2$.
Above $\varepsilon\simeq0.2$ longer words appear, and an exact enumeration of
the primitive zero-sum words up to length $60$ gives the converged values
$K_{\rm tot}=0.976$ ($c=0.244$) at $\varepsilon=1/4$ and $1.689$ ($c=0.270$)
at $\varepsilon=0.3$; for $\varepsilon\ge0.325$ the sum is not converged at
length $60$ (supplementary material, Sec.~\ref{sm:deformation}).

To test Eq.~\eqref{eq:clog} without the heavy-tailed variance, we split
$X_n=X_n^{\rm rest}+X_n^{\rm stay}$ along each orbit, $X^{\rm stay}$ being the
displacement accumulated inside the stays. The centered
$\Var(X^{\rm rest}_n)/n$ is flat over sixteen octaves (slopes in $\ln n$ of
$0.001\pm0.001$ and $0.004\pm0.002$ at $\varepsilon=0.1,0.2$; errors are
bootstrap over orbits here and below), so the whole
logarithmic growth is in the stays [Fig.~\ref{fig:asymmetric}(c)], whose
slope for $n\le2^{18}$ is $0.064\pm0.027$ and $0.215\pm0.079$ against
$0.102$ and $0.230$ from Eq.~\eqref{eq:clog}. At larger $n$ the stays that
matter have length comparable to $n$ and are too rare to sample. With stays
of all even periods up to $14$, $X^{\rm rest}$ is flat also at
$\varepsilon=0.25$ and $0.3$, but not at $0.35$, where longer words matter;
the stay slopes lie systematically below Eq.~\eqref{eq:clog} ($0.15\pm0.05$
against $0.27$ at $\varepsilon=0.3$), so Eq.~\eqref{eq:clog} is tested only to
$\sim40\%$ beyond $\varepsilon=0.25$.

Finally, the deformation moves the zero mean drift to
$\beta=-\alpha(1-2\varepsilon)$. A stay keeps $S$ on the slope $s_m/m$ of its
word; if this equals the mean slope $1-2\varepsilon$, the stay does not
displace $X$, otherwise it does. At $\varepsilon=1/4$ on the zero-drift line
the diffusion is suppressed [$\Var(X_n)/n$ falls from $0.46$ at
$n=2^8$ to $0.2388\pm0.0011$ at $2^{24}$, Fig.~\ref{fig:asymmetric}(a)], whereas for
irrational $\varepsilon=(\sqrt3-1)/3$, where no word has the right slope, it
does not seem to be [a single run gave $\Var(X_n)/n=1.47$ on the line against
$1.46$ at a distance $10^{-3}$ from it]. The triangle map at $\beta=0$ is the case in which the two features
coincide.

\begin{figure}[t]
  \centering
  \includegraphics[width=\textwidth]{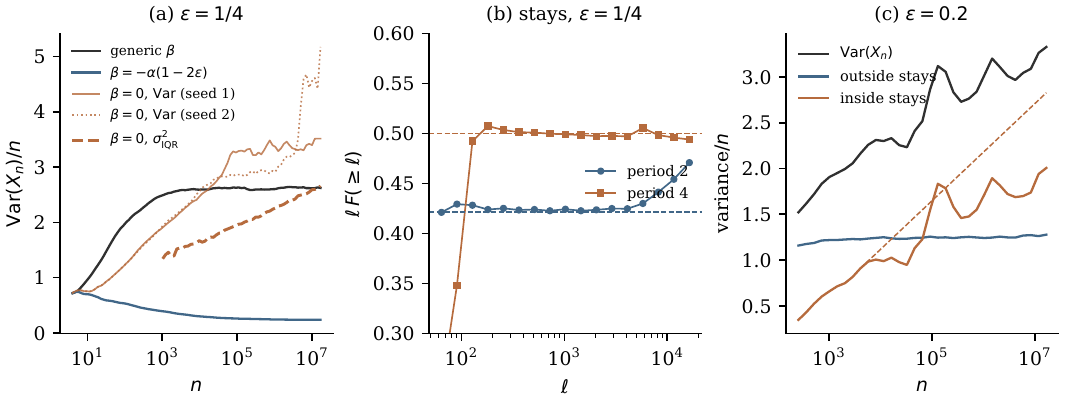}
  \caption{\textbf{Asymmetric deformation, $\alpha=\sqrt2-1$.}
  (a) $\varepsilon=1/4$: $\Var(X_n)/n$ for generic $\beta=\sqrt3-1$, for
  $\beta=-\alpha(1-2\varepsilon)$ (zero mean drift), and for $\beta=0$ (two
  seeds of $10^5$ orbits; interquartile scale $\sigma_{\rm IQR}^2/n$ from
  $n=2^{10}$). (b) $\varepsilon=1/4$, $\beta=0$: $\ell F(\ge\ell)$ for stays
  of period two and four, with the exact $K_2=0.4216$ and $K_4=0.5$ (dashed).
  (c) $\varepsilon=0.2$, $\beta=0$: $\Var(X_n)/n$ split into the parts
  accumulated outside and inside the stays; dashed, the slope
  $c(\varepsilon)=0.230$ of Eq.~\eqref{eq:clog} anchored at $n=2^{12}$.}
  \label{fig:asymmetric}
\end{figure}

\section{Discussion and conjectures}
\label{sec:discussion}

Random-walk scaling does not identify a Bernoulli process, and decorrelation
does not identify diffusion. In the generic triangle map the exponent of the
variance, the Gaussian fluctuations and the first-passage exponent are those
of a coin, but the diffusion coefficient is three times larger and is set by
the order of bounded waiting times, and the persistence amplitude is not
determined by the marginals. At $\beta=0$ the same type of correlations
cancels the Green--Kubo sum, while the walk keeps spreading slowly. The
zero-sum-word resonances, pure twists at $\beta=0$, control the quantities
we can compute: the trapping law, the crossover time $\beta^{-1/2}$ and, in
the deformed map, the logarithmic term~\eqref{eq:clog}. We collect the
picture in three conjectures, in decreasing order of evidence.

\begin{conjecture}[generic map]\label{conj:generic}
If $1$, $\alpha$, $\beta$ are rationally independent, $T_{\alpha,\beta}$ is
uniquely ergodic and mixing, the Green--Kubo series converges to
$D=G_\infty>0$, and, for Lebesgue-distributed initial conditions,
$S_n/\sqrt n$ converges in law to the centered Gaussian of variance $D$.
\end{conjecture}

The evidence is the absence of periodic orbits and the bound on all word
repetitions under rational independence (Sec.~\ref{sec:rigorous}), the
$N^{-1/2}$ decay of worst-case Birkhoff averages, the decay of the Fourier
correlations, and the statistics of Sec.~\ref{sec:generic}: Kolmogorov
distance from the Gaussian at the sampling-noise level for
$10^3\le n\le3.4\times10^7$, running maximum with the Brownian amplitude
$\sqrt{2D/\pi}$ and persistence exponent $1/2$, which also support the
functional version (convergence of $S_{\lfloor nt\rfloor}/\sqrt{Dn}$ to
Brownian motion). The
counterexample of Sec.~\ref{sec:rigorous} shows that rational independence
cannot be weakened to the condition of Ref.~\cite{HorvatEtAl2009}.

\begin{conjecture}[intermediate cohomology at $\beta=0$]\label{conj:cohomology}
For irrational $\alpha$ and $\beta=0$, $T^2$ is ergodic. Consequently, by
Corollary~\ref{cor:coboundary}, $X$ is not a measurable coboundary, the laws
of $S_n$ are not tight, and the spectral measure of $X$ has no atom at $\pi$,
although $\Var(S_n)=o(n)$.
\end{conjecture}

The symbol $X$ would then sit strictly between the two familiar classes of
observables. A coboundary $X=g-g\circ T$ has $S_n=g-g\circ T^n$, whose laws
are tight, and a spectral density vanishing at $\omega=0$; an observable with
$D>0$ has $\Var(S_n)\sim Dn$. At $\beta=0$, $X$ would share $D=0$ with the
first class and unbounded fluctuations with the second: $S_n$ spreads without
limit, but more slowly than $\sqrt n$. The data are consistent: the bulk of the
law of $S_n$ widens, $\Var(S_n)$ shows no plateau up to $10^8$, the
characteristic function~\eqref{eq:char-function} decays, and
$\Var(A_n)/n^2$ decreases, although slowly ($\propto n^{-0.2}$, the same
effective exponent as the Fourier correlations), and the transfer functions of
the exact Markov approximations stay small, growing only as $\ln J$
(Sec.~\ref{sec:beta0}). Whether $S_n/\sigma_n$ has
a non-Gaussian limit, as its stable shape suggests, is a further question.

\begin{conjecture}[resonance renormalization]\label{conj:renormalization}
At $\beta=0$ the growth law of $\Var(S_n)$ is determined by the hierarchy of
zero-sum-word twists, through a renormalization relating the resonances of
words of length $m$ to those of longer words.
\end{conjecture}

Every zero-sum word is a twist at $\beta=0$, resonances exist at all lengths,
a small $\beta$ destroys them on the time scale at which the variance crosses
over to diffusion, and the resonant momenta $(j-\alpha W)/m$ exist for every
$\alpha$, which would explain the insensitivity to the arithmetic of
$\alpha$. In the deformed map the hierarchy provably sets the logarithmic
coefficient through $K_{\rm tot}$; at $\varepsilon=1/2$ the stays carry no
displacement, and the growth must come from the walk between them, which is
the open case. A derivation along these lines would replace the choice among
$\ln^2n$, $(\ln n)^p$ and $n^\gamma$, which the data cannot make.
Renormalization schemes of this kind are established for interval exchange
maps~\cite{Rauzy1979,Veech1982}, for piecewise rotations and polygon
exchanges~\cite{KouptsovLowensteinVivaldi2002,LowensteinKouptsovVivaldi2004,Hooper2013}
and, for the quadratic phase of the lift, for theta sums and Heisenberg
nilflows~\cite{Marklof1999,FlaminioForni2006}; none applies directly here,
because the cylinders of the triangle map are polygons cut by the
discontinuity lines and the block maps are twists rather than translations.

Other questions remain open: the dependence of $D$ on $(\alpha,\beta)$ and the
origin of the maxima of $F(\varepsilon)$; the convergence of $K_{\rm tot}$ for
$\varepsilon>0.3$; and a theory of persistence at $\beta=0$, which the
resonances do not explain, since long-surviving orbits are almost never
repetitions of a periodic word. A further direction is to use observables
adapted to the finite-time continuity domains and their successive cuts,
rather than global Fourier modes, for which the recurrence analysis of the
cuts in Ref.~\cite{DegliEspostiGalatolo2005} is a natural starting point.

\section*{Supplementary Material}
The supplementary material contains a table of what each diagnostic can and
cannot establish, the word statistics and the twist geometry at $\beta=0$,
further data for the generic regime (correlations, other parameter pairs and
the map of $D(\alpha,\beta)$), the $\beta=0$ variance fits and Markov
approximations, the worst-case Birkhoff averages, details of the asymmetric
deformation, and the phase-space coverage tests.

\section*{Author Declarations}
\subsection*{Conflict of Interest}
The author has no conflicts to disclose.

\subsection*{Author Contributions}
\textbf{Mirko Degli Esposti}: all roles (Conceptualization, Methodology,
Software, Formal analysis, Investigation, Validation, Visualization,
Writing -- original draft, Writing -- review \& editing).

\section*{Data availability}
The data that support the findings of this study and the code used to
generate them are openly available in Zenodo at
\url{https://doi.org/10.5281/zenodo.23623347}~\cite{DegliEsposti2026code}.

\appendix

\section{Numerical methods}
\label{app:methods}

\paragraph{Integrator.} The integers $S_n$ and $W_n$ of
Eq.~\eqref{eq:exact-phase} are updated exactly; $q_n$ is recomputed at every
step from $q_0$, $p_0$ (double precision) and $\alpha$, $\beta$ (double-double,
$\sim10^{-32}$), with an absolute error of order $10^{-16}$ that does not
accumulate, so $X_n$ can be wrong only if $q_n$ falls within this distance of
a cut. The construction requires $|W_n|$, $n(n+1)/2<2^{53}$, i.e.\
$n\lesssim10^8$. Symbol sequences coincide with $60$-digit \texttt{mpmath}
orbits over $2\times10^4$ steps, at $\beta=0$ and for generic parameters, and the cost is $1.3$ times that of double-precision iteration.

\paragraph{Polygonal cylinders.} For a fixed word, $q_k$ is affine in
$(q_0,p_0)$, so each cylinder is a finite union of convex polygons. The
partition into cylinders of length $n$ is built by clipping with the lines
$q_0+kp_0=\text{const}$, $k<n$ (pieces below $10^{-16}$ in area dropped; lost mass
$<10^{-13}$). Then $\Var(S_n)=\sum_w|[w]|\,s_n(w)^2$, with $|[w]|$ the area of $[w]$, is a finite sum of
polygon areas, and stationarity gives
$C(k)=[\Var(S_{k+1})-2\Var(S_k)+\Var(S_{k-1})]/2$. The number of pieces grows as $\simeq0.7n^3$, of the order of
the cubic growth reported in Ref.~\cite{HorvatEtAl2009}, which limits exact
computations to $n\lesssim100$ (to $n=256$ for the survival probability $Q_+(n)$, for which only the
cylinders with $S_k>0$ for all $k$ are refined). Word
probabilities, $\Var(S_n)$, $C(k)$ and $Q_+(n)$ obtained in this way agree
with Monte Carlo estimates. The trap coefficients $K_w$ are computed exactly
for all primitive zero-sum words of length $m\le60$ by a depth-first search
over intervals of $q$ on each resonant line $p=(j-\alpha W)/m$.

\paragraph{Correlations and spectra.} Correlations are estimated per
trajectory with the exact mean $0$, by FFT on zero-padded segments, and
averaged over independent trajectories. Spectra are averaged periodograms of
Hann-windowed segments of lengths $2^{15}$ to $2^{24}$ with $50\%$ overlap,
binned logarithmically and combined across segment lengths; the lowest
frequency is $\omega\simeq3\times10^{-6}$.

\paragraph{Statistics.} The statistical unit is always an independent
initial condition drawn from Lebesgue measure. Errors on variances are
bootstrap or batch standard deviations over independent orbits; survival probabilities carry
Wilson $68\%$ intervals; effective exponents are computed over windows of
fixed ratio. Replicate ensembles with independent seeds agree within errors, except for
the heavy-tailed raw variances of Sec.~\ref{sec:asymmetric}, for which we
report robust scales.
The same pipeline applied to Bernoulli sequences reproduces $D=1$ and
$S_X(0)=1$ within $4\%$ (supplementary material, Sec.~\ref{sm:generic}).

\paragraph{Use of AI tools.} The simulation and analysis code was written
and refactored with the assistance of Claude Opus~5.5 and Claude Sonnet~5.5
(Anthropic), used through the Claude Code agent, to speed up code
development and to obtain independent critical readings of the drafts. The
same tools proposed
numerical tests, drafted and revised passages of the manuscript, produced
simulated referee reports on earlier drafts, assisted with the literature
search, and checked quantitative statements in the text against the
project's own logs. The author reviewed and edited all content, checked the
mathematical arguments, and takes full responsibility for the article.
Every bibliographic entry was checked against the publisher's or a
bibliographic database record, and every numerical result reported here is
produced by the deposited code~\cite{DegliEsposti2026code}.

\bibliographystyle{unsrt}
\bibliography{references}

\ifarxiv
\clearpage
\begin{center}
{\Large\bfseries Supplementary Material}
\end{center}

\setcounter{section}{0}
\setcounter{figure}{0}
\setcounter{table}{0}
\setcounter{equation}{0}
\renewcommand{\thesection}{S\arabic{section}}
\renewcommand{\thefigure}{S\arabic{figure}}
\renewcommand{\thetable}{S\arabic{table}}
\renewcommand{\theequation}{S\arabic{equation}}

Experiment numbers refer to the scripts \texttt{experiments/NN\_*.py} of the
code archive (see Data availability); they follow the chronological research
log, so not every number appears here.
Table~\ref{tab:diagnostics} summarizes what each diagnostic of the main text
can and cannot establish.

\begin{table}[ht]
\centering
\small
\begin{tabular}{p{0.24\textwidth}p{0.33\textwidth}p{0.33\textwidth}}
\toprule
Diagnostic & What it measures & What it cannot establish\\
\midrule
$\Var(S_n)$, $C(k)$, $G_K$, $S_X(\omega)$ & second-order structure; $D$ and finite-size corrections on the explored scales & asymptotic limits; any non-Gaussian or higher-order property\\
CLT diagnostics & one-time distribution of $S_n$, two-time increments & a functional limit theorem\\
First passages $Q_\pm(n)$ & persistence, a multi-time property & a limiting exponent beyond the explored window\\
Inter-event intervals & distribution versus order of the waiting times & the asymptotics of the order effect\\
Exact cylinders, Markov hierarchy & word probabilities without sampling; finite-order models & the asymptotic behavior of the map\\
Fourier modes, worst-case averages & decay of correlations; uniformity of Birkhoff averages & mixing or unique ergodicity\\
\bottomrule
\end{tabular}
\caption{Diagnostics used in this work and their scope. Variance,
correlations and spectrum are three estimators of the same second-order
structure; their agreement is a consistency check rather than independent
evidence.}
\label{tab:diagnostics}
\end{table}

\section{Words, runs and exact probabilities}
\label{sm:words}

\paragraph{Run lengths near rational drifts (Experiment 21).} The bound of
Lemma~\ref{lem:runs} is $L_s=L(c_s)$, $c_s=\beta+s\alpha$, where $L(c)$ is the
longest run of the skew shift $F_c$ inside a half-torus; it is the same for
both halves, since $q\mapsto q+1/2$ commutes with $F_c$ and exchanges them, and
$L(c)=L(-c)=L(1-c)$ by the involution $(q,p)\mapsto(-q,-p)$. Hence
$L_+(\alpha,\beta)=L_-(\alpha,-\beta)$, and $L_+=L_-$ at $\beta=0$, but not in
general. Figure~\ref{fig:runs} shows $L(c)$, computed exactly from the
polygonal run cylinders for $20\,000$ values of $c$. It is piecewise constant
and small ($6$--$7$ on most of $[0.1,0.3]$) and diverges only at the
rationals $a/b$ for which $F_{a/b}$ has a periodic orbit contained in one
half-torus: among those of small denominator, $c=0$, $1/3$, $2/5$ and $1/2$,
but not $1/5$, $1/4$, $2/7$ or $3/7$. Near them $L\propto|c-a/b|^{-1/2}$
(local exponents $0.41$--$0.49$ on $10^{-4}\le|c-a/b|\le10^{-2}$), which for
$c\to0$ reflects $q_k-q_0\simeq kp_0+ck^2/2$. In the $(\alpha,\beta)$ plane
the resonances form the two families of diagonals $\beta\pm\alpha=a/b$. For
the generic benchmark, $c_+=\beta+\alpha\simeq0.986$ is close to $c=0$ and
$L_+=17$, whereas $c_-\simeq0.250$ gives $L_-=7$. Irrationality of $\alpha$,
$\beta$ with $\beta\notin\mathbb Z+\alpha\mathbb Z$ does not make both $c_\pm$
irrational: $\beta=\tfrac12-\alpha$ gives $c_+=\tfrac12$.

\begin{figure}[ht]
 \centering
 \includegraphics[width=\textwidth]{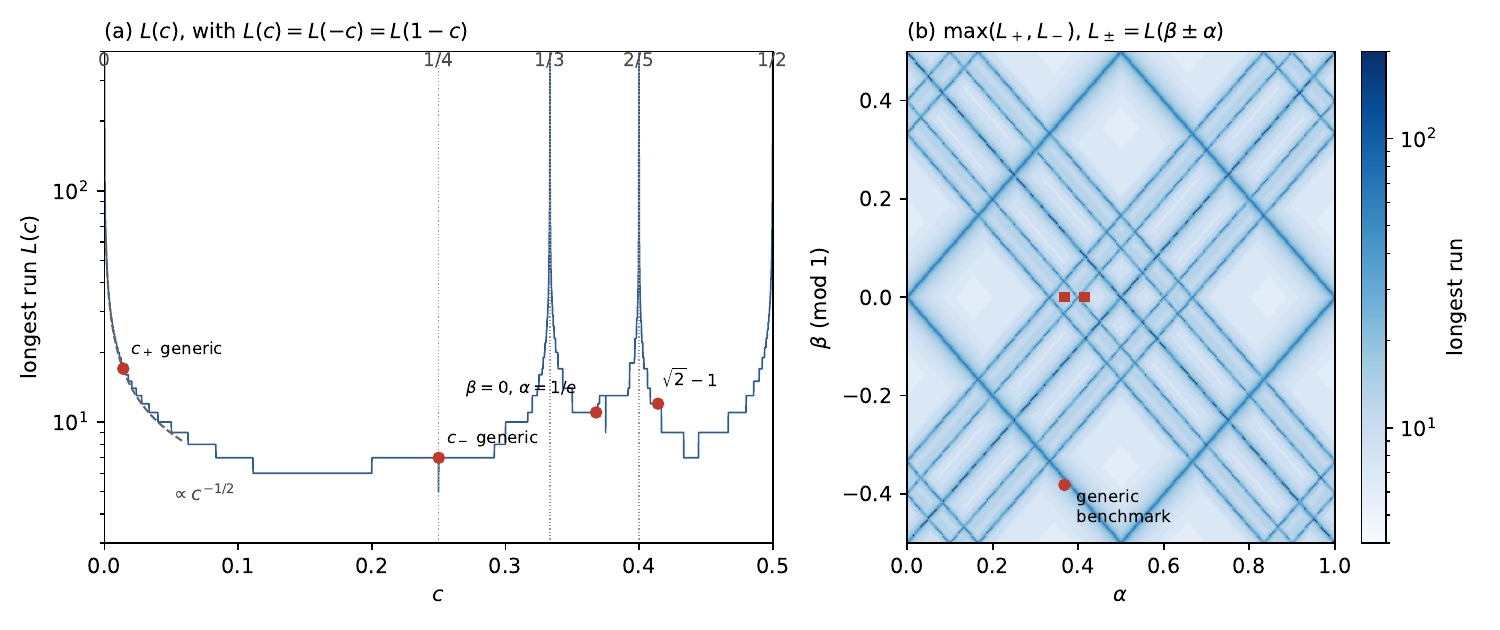}
 \caption{\textbf{The run bound of Lemma~\ref{lem:runs}} (Experiment 21).
 (a) Longest run $L(c)$ of the skew shift $F_c$ in a half-torus, computed
 exactly for $20\,000$ values of $c\in(0,1/2]$; dotted lines mark rationals,
 dashed the law $c^{-1/2}$ near $c=0$. Dots: $c_\pm$ of the generic benchmark
 and $c=\alpha$ for $\beta=0$ with $\alpha=1/e$, $\sqrt2-1$. (b)
 $\max(L_+,L_-)$ with $L_\pm=L(\beta\pm\alpha)$ in the $(\alpha,\beta)$ plane;
 the dot is the generic benchmark ($\beta\equiv g-1$), the squares the two
 $\beta=0$ cases.}
 \label{fig:runs}
\end{figure}

\paragraph{The period-two twist and its traps.} For $w=+-$ at $\beta=0$, a
point with $q\in[0,1/2)$ and $q+p+\alpha\in[1/2,1)\pmod1$ is mapped by $T^2$ to
$(q+2p+\alpha,\,p)$: on the cylinder $[+-]$ the block map is a pure twist,
the identity on the lines $2p+\alpha\in\mathbb Z$. Imposing the symbols gives
the two segments of Fig.~\ref{fig:twist}(a), $p=1-\alpha/2$,
$q\in[\tfrac12-\tfrac\alpha2,\tfrac12)$ and $p=\tfrac12-\tfrac\alpha2$,
$q\in[0,\tfrac12-\tfrac\alpha2)$, together with their images under one step,
which carry the rotation $-+$. At momentum distance $\delta$ the block map is
the translation $q\mapsto q+2\delta$, so an orbit slides along the segment
and repeats $+-$ for the time it needs to cross it, $L/|\delta|$ steps
[Fig.~\ref{fig:twist}(b): $1415$, $567$, $944$ steps against $1414$, $566$,
$943$]. For $\beta>0$ the block map adds $2\beta$ to the momentum, so
$\delta$ grows linearly in time and the motion along the segment becomes a
parabola; the longest repetition is the one whose parabola, of depth
$\beta n^2/8$ after $n$ steps, just fits in the longer segment, of length
$L_{\max}=\max(\alpha,1-\alpha)/2$, i.e.\ $n=2\sqrt{2L_{\max}/\beta}$
[Fig.~\ref{fig:twist}(c): $483$ steps against $484$ at $\beta=10^{-5}$].

\begin{figure}[ht]
 \centering
 \includegraphics[width=\textwidth]{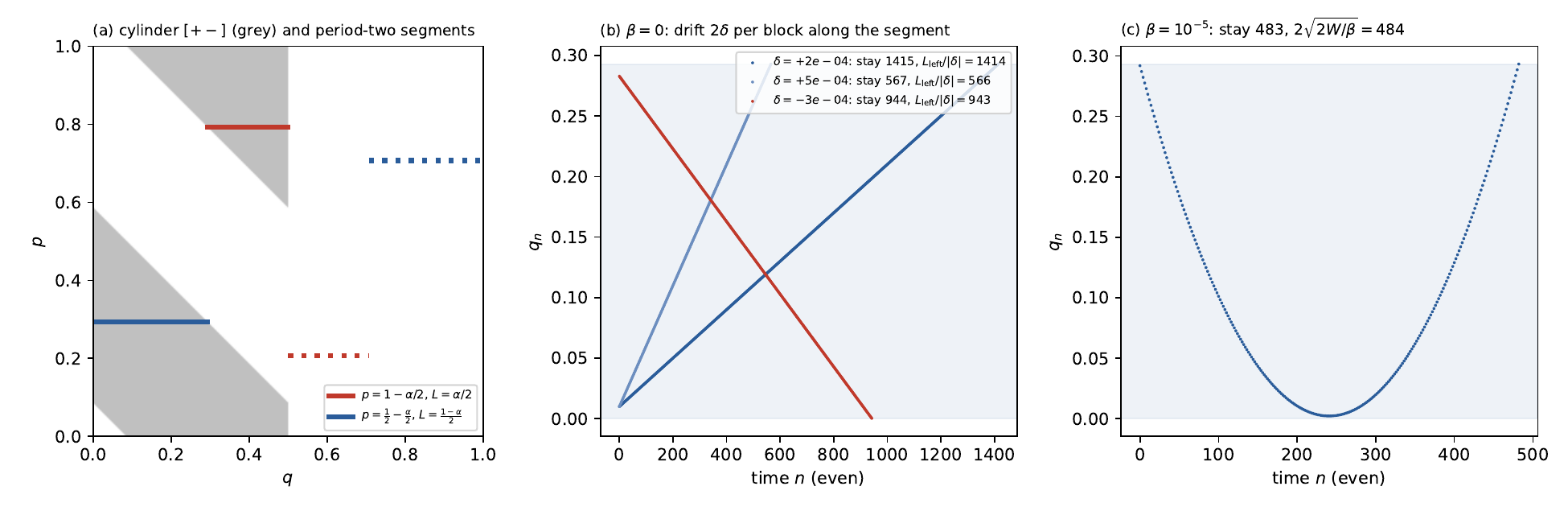}
 \caption{\textbf{The period-two twist at $\beta=0$}, $\alpha=\sqrt2-1$.
 (a) The cylinder $[+-]$ (grey), the period-two segments (solid) and their
 images under one step (dotted). (b) $q_n$ at even $n$ for orbits at momentum
 distance $\delta$ from the segment $p=\tfrac12-\tfrac\alpha2$ (shaded): they
 drift by $2\delta$ per block and repeat $+-$ until they leave it; the stay
 length is compared with the remaining length of the segment divided by
 $|\delta|$. (c) $\beta=10^{-5}$: the momentum drift bends the motion into a
 parabola; the longest stay fits it into the wider segment.}
 \label{fig:twist}
\end{figure}

\paragraph{$P(-+-+)=1/12$ at $\beta=0$.} With $r=q_0$ and $x=p_0-\alpha$
(independent and uniform),
$P(-+-+)=\int_0^1\!dx\int_0^1\!dr\,\mathbf 1_{A_x}(r)\mathbf 1_{A_x}(r+2x+\alpha)$,
where $A_x$ is an arc of length $\|x\|$. Its Fourier coefficients in $\alpha$
are proportional to $\int_0^1\sin^2(\pi k\|x\|)e^{4\pi ikx}dx=0$ for every
integer $k\ne0$, so $P(-+-+)=\int_0^1\|x\|^2dx=1/12$ for every $\alpha$. The
identity was first found numerically, to $10^{-10}$, in the exact cylinder
computations of Sec.~\ref{sm:markov}; it fails for $\beta\ne0$.

\paragraph{Word complexity (Experiment 20).} From the exact cylinder
partition we count, for $n\le110$, the distinct admissible words $p(n)$ (a
cylinder made of several polygons is counted once), the number of polygons,
and the admissible zero-sum words $z(n)$, i.e.\ those with $S_n=0$
(Table~\ref{tab:complexity}). All words are admissible up to $n=5$ (generic)
and $n=6$ ($\beta=0$). The first difference satisfies $s(n)/(3n^2)=0.661$,
$0.704$, $0.708$ at $n=110$, in agreement with $p/n^3$: the growth is cubic,
with coefficient $\simeq0.66$ (generic) and $\simeq0.70$--$0.71$ ($\beta=0$,
still rising slowly). The second difference $b(n)$ gives $b/(6n)$
fluctuating between $0.50$ and $0.93$, and cubic fits on $n\ge30,50,70$ give
$0.57$--$0.74$, so neither the cubic coefficient nor the lower-order terms are
determined. The number of polygons exceeds $p(n)$ by less than $0.03\%$:
nearly every cylinder is connected. The topological entropy vanishes. At
$\beta=0$ the zero-sum words are a constant fraction $\simeq23\%$ of all
words, $z(n)\propto n^{3.0}$, against $z(n)\propto n^{2.1}$ and a falling
fraction ($4\%$ at $n=110$) in the generic case, where the spread $\sqrt n$
of $S_n$ alone would give $n^{2.5}$. This is consistent with the zero-sum
twists of Sec.~\ref{sec:bounded-runs}, but we have no derivation.

\begin{table}[ht]
\centering
\caption{Word complexity at $n=110$ (Experiment 20). $p$: distinct words;
$z$: distinct words with $S_n=0$.}
\label{tab:complexity}
\begin{tabular}{lrrrrr}
\toprule
case & $p$ & polygons & $p/n^3$ & $z$ & $z/p$\\
\midrule
generic ($\alpha=1/e$, $\beta=(\sqrt5-1)/2$) & 881\,378 & 881\,620 & 0.662 & 34\,356 & 0.039\\
$\beta=0$, $\alpha=\sqrt2-1$ & 938\,742 & 938\,942 & 0.705 & 215\,138 & 0.229\\
$\beta=0$, $\alpha=1/e$ & 943\,784 & 943\,836 & 0.709 & 217\,263 & 0.230\\
\bottomrule
\end{tabular}
\end{table}

\section{Generic regime: further data}
\label{sm:generic}

\paragraph{Long runs.} With $4096$ initial conditions to $n=10^8$ the
late-time exponents are $\gamma_{\rm var}=0.9957$ and $\gamma_{\max}=0.5024$,
with $\Var(S_n)/n=3.07854$ and $\langle S_{\max}(n)\rangle/\sqrt n=1.4077$ at
the largest checkpoint; the run to $n=10^7$ with $2048$ orbits gave $1.0086$,
$0.5062$ and $3.10155$. The ensemble variance of a few thousand orbits has a
relative error of about $3\%$ at every $n$, whereas long-trajectory
correlations determine $D_n$ for $n\lesssim10^4$ ten to forty times more
precisely. An earlier truncated Green--Kubo value $3.425$ at
$K=2\times10^5$, from eight double-precision trajectories, is not confirmed
($G_{1.3\times10^5}=3.08\pm0.04$). Figure~\ref{fig:correlations} shows the
individual correlations: the measured $C(k)$ coincide with the exact cylinder
values for $k<80$, with deviations distributed as statistical noise, and the
RMS correlation falls by two decades across the band
$10^2\lesssim k\lesssim2\times10^3$ before reaching the noise level
$\simeq2\times10^{-5}$.

\begin{figure}[t]
 \centering
 \includegraphics[width=\textwidth]{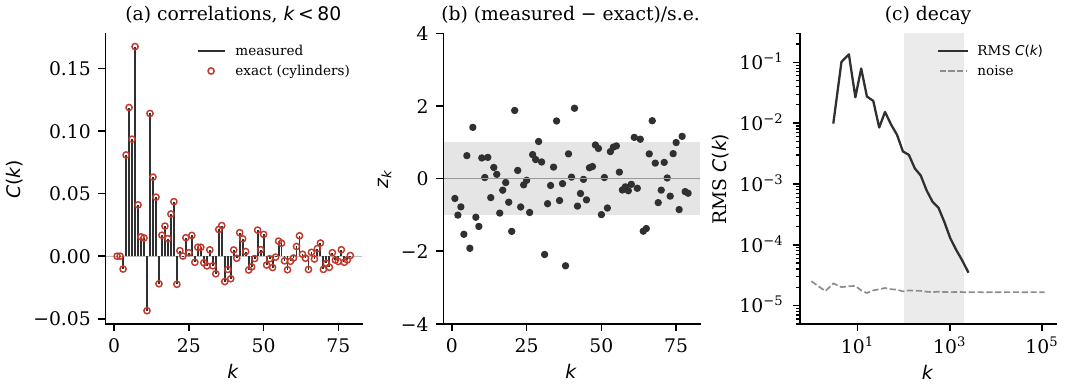}
 \caption{\textbf{Symbol correlations, generic regime} ($\alpha=1/e$,
 $\beta=(\sqrt5-1)/2$; $64$ exact-phase trajectories of length $2^{26}$).
 (a) $C(k)$ for $k<80$: measured (bars) and computed without sampling from
 the polygonal cylinders (circles); $C(1)=C(2)=0$. (b) Deviations in units
 of the standard error. (c) RMS of $C(k)$ in logarithmic bins of $k$ with the
 noise level subtracted (solid) and the noise level (dashed); grey: the band
 $10^2\le k\le2\times10^3$.}
 \label{fig:correlations}
\end{figure}

\paragraph{Finite-size corrections.} The first moment $M=\sum_kkC(k)$
converges on the accessible scales ($M\simeq95$--$100$ for
$2\times10^3\le K\le10^4$), so $D_n=D-2M/n+o(1/n)$ reaches $D$ to $1\%$ only
for $n\gtrsim2\times10^4$. A non-oscillating tail $C(k)\simeq ak^{-3/2}$ would
give $D_n-D\propto n^{-1/2}$, $G_K-D\propto K^{-1/2}$ and a cusp
$S_X(\omega)-D\propto\sqrt{|\omega|}$. A fit restricted to $k<10^3$ suggests
$a\approx0.6$--$1$, but over $2\times10^3\le k\le10^4$ the data give
$a=0.08\pm0.16$. The $k^{-3/2}$ decay reported for the smooth observable
$\cos\pi y$~\cite{CasatiProsen2000} therefore does not carry over to
$\thetafun$ on these scales. The periodogram decreases roughly linearly in
$\sqrt\omega$ for $10^{-3}\lesssim\omega\lesssim10^{-1}$, the image of the
correlation band at $k\sim10^2$--$10^3$, and is flat below.

\paragraph{Other generic pairs (Experiment 12).} For
$(\alpha,\beta)=(\sqrt2-1,g)$, $(1/e,\sqrt2-1)$, $(\pi-3,1/e)$ and
$(g,\sqrt3-1)$, $g=(\sqrt5-1)/2$, Green--Kubo sums from $32$ trajectories of
length $2^{24}$ give $D=2.19$, $2.94$, $2.15$ and $2.33$ (errors $\simeq0.03$;
the ensemble variance at $n=2^{20}$ agrees). The ratio of $D$ to the renewal
value $v/2$ is $2.4$--$3.7$.

\paragraph{Map of $D(\alpha,\beta)$ (Experiment 24).} Besides $\alpha\to\alpha+1$,
$\alpha\to-\alpha$ ($q\to q+\frac12$) and $\beta\to-\beta$ ($\sigma$), the identity
$\theta/2\equiv\frac12\pmod1$ gives $T_{\alpha+1/2,\beta}=T_{\alpha,\beta+1/2}$, hence
$D(\alpha,\beta)=D(\frac12-\alpha,\frac12-\beta)$: the fundamental domain is the
triangle $\alpha,\beta\ge0$, $\alpha+\beta\le\frac12$. Fig.~\ref{fig:Dmap} shows
$D_n=\Var(S_n)/n$ at $n=2^{14}$ on a $200\times200$ grid with irrational offsets
($1024$ orbits per point, statistical error $\simeq4\%$); $D_n$ changes by
$\simeq5\%$ between $n=2^{10}$ and $2^{14}$, except at $\beta\simeq0$ and near
$\alpha=0,\frac12$ ($0.4\%$ of the points change by more than $50\%$). Away from the
edges ($\alpha,\beta\ge0.02$, $\alpha\le0.48$) $D$ ranges from $0.20$ to $14$
(1st--99th percentiles $0.7$--$6.8$, median $2.3$) and is below the Bernoulli
value on $6\%$ of the points; towards the skew-shift limit $\alpha\to0$ and near
$\beta=0$, $D_n$ is still growing at $n=2^{14}$ (up to $53$ at
$\alpha\simeq\beta\simeq10^{-3}$) and is not a converged value. Dark lenses, with $D$ well below $1$, sit at rational $\beta=j/m$ with
small $m$ and open from the skew-shift limit $\alpha=0$ [panel (b)]: there
zero-sum words whose length is a multiple of $m$ have $c_w\in\mathbb Z$, the
mechanism that gives $D=0$ at $\beta=0$. Bright bands follow the low-order
resonances $m\beta+s\alpha\in\mathbb Z$ with $s\ne0$ (e.g.\ $\beta=\alpha$ and
$\alpha+\beta=\frac12$) near the corners of the square but leave them in the
interior.

\begin{figure}[ht]
  \centering
  \includegraphics[width=\textwidth]{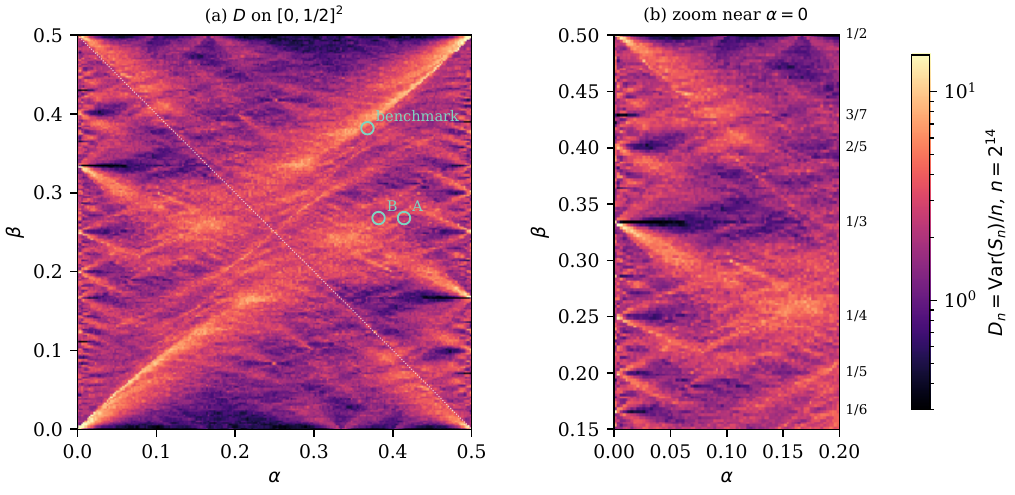}
  \caption{\textbf{Diffusion coefficient over the parameter square} (Experiment 24,
  exact lift). $D_n=\Var(S_n)/n$ at $n=2^{14}$, $1024$ orbits per point, on a
  $200\times200$ grid; computed for $\alpha+\beta\le\frac12$ and reflected through
  $(\frac14,\frac14)$ (dotted line), by the exact symmetry
  $D(\alpha,\beta)=D(\frac12-\alpha,\frac12-\beta)$. (a) The fundamental square
  $[0,\frac12]^2$, with the benchmark pair and the pairs A, B of
  Sec.~\ref{sec:asymmetric} folded into it. (b) Zoom near the skew-shift limit
  $\alpha=0$; labels mark rational $\beta$.}
  \label{fig:Dmap}
\end{figure}

\paragraph{Distribution of $S_n$.} The skewness of $S_n$ is $0.65$--$0.77$
for $n\le30$, $0.41$ at $n=10^2$ and $0.19$ at $n=10^3$, decaying roughly as
$n^{-1/2}$ and compatible with zero from $n\sim10^4$. The involution
$(q,p)\mapsto(-q,-p)$ maps $S_n$ for $(\alpha,\beta)$ to $-S_n$ for
$(\alpha,-\beta)$, so the sign of the skewness is reversed by
$\beta\mapsto-\beta$; for the benchmark $S_n$ is slightly more often
negative, consistent with $Q_+>Q_-$. The correlation between $S_n$ and
$S_{2n}-S_n$ is $0.07\pm0.02$ at $n=10^2$.

\paragraph{Bernoulli control.} The same finite-sample pipeline applied to
independent unbiased Bernoulli sequences, on the same scales, gives
$\Var(S_n)/n=0.9676$, $G_{200000}=0.9856$, and a low-frequency spectrum
$S_X(0)\simeq1.015$ with local exponent $\eta=0.0069$, consistent with $D=1$
and $S_X\equiv1$. With few trajectories the Green--Kubo partial sum is more
sensitive to accumulated noise than the spectral estimate.

\section{The $\beta=0$ regime: further data}
\label{sm:beta0}

\paragraph{Precision.} The earlier study of this case already used
extended precision~\cite{HorvatEtAl2009}. On the same initial conditions,
double-precision symbol sequences first differ from the exact ones at
$n\sim10^6$--$7\times10^6$ at $\beta=0$ (at $n\sim8\times10^4$ in the generic
regime). For most $\alpha$ the variances still agree within the sampling
noise, and a double-precision run with $4096$ orbits gives
$\Var(S_{10^7})=88.0$ and $\Var(S_{10^8})=112.6$, consistent with the exact
data; for $\alpha=[0;2,30,2,300,2,3000,1,1,\dots]$, however, the
double-precision variance is three times too large for $n\gtrsim5\times10^6$.

\paragraph{Growth law.} The six values of $\alpha$ and the replicate
ensembles give the same curves within errors (about $5\%$ at the largest
scale); $\alpha\simeq0.49194$ shows a short transient at $n\sim20$--$60$,
reflecting $\alpha\simeq1/2$, and no feature is visible at the convergent
denominators $33102$ and $74646$ of $\pi-3$. The models $a+b\ln^2n$,
$a+c\ln n+b\ln^2n$, $a+b(\ln n)^p$ and $a+cn^\gamma$ fit all ensembles equally
well for $n\ge10^2$, with $b\simeq0.2$--$0.35$, $p\simeq1.5$--$2.1$ and
$\gamma\simeq0.05$--$0.1$ depending on the model and on $\alpha$. They are
numerically degenerate over seven to eight decades
($n^\gamma=1+\gamma\ln n+\tfrac12\gamma^2\ln^2n+\dots$), the coefficient of
$\ln^2n$ is not a robust quantity, and the offset $a\simeq6$--$8.5$ differs
between independent ensembles. The logarithmic growth of the extremes
$S_{\max}(n)$ of single orbits reported for quadratic irrational $\alpha$ in
Ref.~\cite{HorvatEtAl2009} concerns a different quantity and is not in
conflict with these results. On the explored scales we cannot exclude a very
slow crossover to diffusion with a tiny $D$, or a small power $\gamma$.

\paragraph{Spectrum.} Binned values of $1-\eta(\omega)$ scatter by about
$\pm0.3$, with decade averages mostly between $0$ and $0.2$; the drift
$\eta=1.004\to0.861$ seen in double-precision runs of different sizes is
within this scatter. The variance reconstructed from the measured spectrum
through Eq.~\eqref{eq:spectral-diffusion} reproduces the shape of the
ensemble variance, $10$--$15\%$ below it. The residual oscillations of $G_K$
are of arithmetic origin rather than noise; $A(1.3\times10^5)=\sum_{k\le K}C(k)=-0.494\pm0.001$.

\paragraph{Parity and the $\pi$ peak.} For $\alpha=1/e$, $\Var(A_n)/n^2$
decreases from $0.034$ to $0.0061$ between $n=2^{10}$ and $2^{22}$, and the
peak of $S_X$ at $\pi$ has height $\simeq680$. There are no peaks at $2\pi/5$
or at irrational frequencies. Per orbit, $A_n/n$ is nearly constant over
decades, with abrupt jumps. At rational frequencies other than $\pi$ the
phases of successive bursts are not locked, which is why the $\pi$ peak
dominates although the trap coefficients of the length-four words are
larger. The trap law~\eqref{eq:trap-law} for $\alpha=1/e$ agrees with the
exact coefficients within $5$--$20\%$; at $\beta=0$ itself the estimates
scatter more than at $\beta=10^{-7}$ because each orbit's momentum lattice is
frozen.

\paragraph{Persistence by momentum class (Experiment 10).} At $\beta=0$,
$p_n-\alpha S_n=p_0\pmod1$, so $p_0$ labels families of orbits that evolve
independently on $\mathbb T\times\mathbb Z$. Sampling $p_0$ on $1024$ bins
($128$ orbits each, $\alpha=\sqrt2-1$): the growth of $\Var(S_n)$ is shared by
all classes (at $n=10^6$ the $10\%$ of classes with the largest conditional
variance carry $15\%$ of the total), whereas long survival is carried by
thin, $q_0$-independent bands of $p_0$: at $n=10^6$, $77\%$ of the classes
have no survivor, against $28\%$ expected for a $p_0$-independent survival.
The bands are selected early, $P(\tau_+>256\mid p_0)$ ranging from $0.003$
outside to $0.2$--$0.3$ inside, and this profile is reproduced exactly by the
polygonal cylinders pruned when $S_n$ becomes positive. Some band edges
coincide with the resonances $\lVert2p_0+(2k-1)\alpha\rVert=0$ of orbits
bouncing between two adjacent levels. After the selection the dynamics is
common to all classes: conditional on survival to $4096$, the probability
of surviving to $2^{20}$ is $0.25$--$0.35$ in almost every active class,
corresponding to $\theta\simeq0.22$.

\paragraph{Persistence: null models.} For fractional Brownian motion the
persistence exponent is $1-H$~\cite{Molchan1999}; the local variance
exponents $2H\simeq0.1$--$0.2$ would give $\theta\simeq0.9$, far from the
measured $0.2$. The Sparre Andersen identity fed with the measured
marginals fails at $\beta=0$ in the exponent, not only in the amplitude. For
the generic case its extensions to correlated steps exist only under
renewal-type assumptions~\cite{ArtusoCristadoroDegliEspostiKnight2014}; we
use it only as a null model, which reproduces simulated Bernoulli walks
within $1\%$ up to $n=4096$.

\section{An exact Markov hierarchy}
\label{sm:markov}

Let $\Sigma=\{\thetafun(q_0)=+1,\ \thetafun(q_0-p_0)=-1\}$ be the set of points
where a run of $+1$ begins; $\mu(\Sigma)=1/4$ exactly and the mean return
time is $4$. By Lemma~\ref{lem:runs} the first-return map $R$ to $\Sigma$ is
piecewise affine with finitely many branches, one for each cycle
$(+1)^a(-1)^b$ ($67$ generic, $58$ for $\beta=0$, $\alpha=1/e$, $63$ for
$\beta=0$, $\alpha=\sqrt2-1$); on each branch
$R(q,p)=(q+np+\beta n(n+1)/2+\alpha W,\;p+n\beta+\alpha(a-b))$, $n=a+b$. The
increment of $S$ over a cycle is $\varphi=a-b$, and $D=\sigma^2(\varphi;R)/4$.
Cylinders of $R$ of length $J+1$ are obtained by intersecting affine images
of branches with the branch partition, and define the order-$J$ Markov chain
on cylinders of length $J$, in the spirit of the Markov approximations of
Ref.~\cite{HorvatEtAl2009} but without sampling. Its asymptotic variance is
$\sigma_J^2=2\langle\varphi,x\rangle_\pi-\langle\varphi,\varphi\rangle_\pi$,
where $x$ solves the Poisson equation $(I-P)x=\varphi$ (sparse direct solver),
and $D_J=\sigma_J^2/4$. For $J\le28$ ($\simeq2.7\times10^5$ states)
$D_J\simeq2.9$ in the generic regime, while for $\beta=0$
$D_J\simeq0.78/J$ ($14\le J\le28$) for both values of $\alpha$. Since the
chain retains a memory of about $4J$ steps, $4JD_J\simeq3.1$ should be
compared with $d\Var(S_n)/d\ln n\simeq3.2$--$3.7$ at $n\simeq10^2$. The
transfer function $x$ has root-mean-square size growing like $\sqrt J$ in the
generic regime but only logarithmically at $\beta=0$ ($x_{\rm rms}\simeq3$
at $J=28$), suggestive of a near-coboundary structure. A finite-order chain
does not determine the asymptotic behavior of the map.

\section{Worst-case Birkhoff averages}
\label{sm:worst}

Unique ergodicity is equivalent to $\sup_x|N^{-1}\sum_{n<N}f(T^nx)|\to0$ for
a dense set of continuous $f$ (here the Fourier modes, with the usual care
for the discontinuity of $T$), which typical-orbit statistics do not probe.
We estimate the supremum by the maximum over $32\,768$ Lebesgue initial
conditions, over $8192$ initial conditions within $10^{-14}$--$10^{-3}$ of
the cuts, and by a greedy adversarial search at $N=2^{10},2^{14},2^{16}$
($64$ parents $\times$ $63$ children, radii $10^{-2}$--$10^{-7}$), whose worst
points are then followed to $N=2^{20}$ (Experiment 15, exact lift).

\begin{table}[ht]
\centering
\caption{Largest $|N^{-1}\sum_{n<N}e^{2\pi iq_n}|$ found
($g=(\sqrt5-1)/2$).}
\label{tab:worst}
\resizebox{\linewidth}{!}{\begin{tabular}{lcccc}
\toprule
 & $N=2^{10}$ & $2^{14}$ & $2^{18}$ & $2^{20}$\\
\midrule
generic $(1/e,g)$, random max / median & 0.179 / 0.035 & 0.049 / 0.009 & 0.012 / 0.0023 & 0.0062 / 0.0011\\
generic $(1/e,g)$, near cuts & 0.205 & 0.046 & 0.011 & 0.0051\\
generic $(1/e,g)$, adversarial at $2^{10}$ & 0.213 & 0.019 & 0.0073 & 0.0026\\
generic $(\sqrt2-1,g)$, random max & 0.274 & 0.068 & 0.011 & 0.0050\\
generic $(\sqrt2-1,g)$, adversarial at $2^{10}$ & 0.293 & 0.042 & 0.0080 & 0.0035\\
skew shift $(0,g)$, random max & 0.060 & 0.015 & 0.0035 & 0.0016\\
$\beta=0$ $(\sqrt2-1)$, random max / median & 0.795 / 0.097 & 0.606 / 0.075 & 0.587 / 0.061 & 0.451 / 0.055\\
$\beta=0$, adversarial at $2^{14}$ & 0.796 & 0.796 & 0.796 & 0.796\\
counterexample $(1/e,(1-\alpha)/2)$, adversarial & 0.425 & 0.425 & 0.425 & 0.425\\
\bottomrule
\end{tabular}}
\end{table}

For both generic pairs the worst case decays as $N^{-1/2}$ (exponent $0.49$
over three decades), like the uniquely ergodic skew shift, with a ratio
max/median $\simeq5$--$6$ that is stable in $N$ (it would be $\simeq4$ for Gaussian
fluctuations): the worst case is the extreme of a fixed, somewhat heavy-tailed
distribution of $\sim3\times10^4$ diffusive fluctuations. Adversarially selected points are transients and
become typical by $N=2^{20}$; points next to the cuts behave like random
ones. The method finds the known obstructions: at $\beta=0$ it recovers both
period-two segments, $p=1-\alpha/2$ with average $\cos(\pi\alpha/2)=0.7957$ and
$p=\tfrac12-\tfrac\alpha2$ with $\cos(\pi(1-\alpha)/2)=0.6057$, constant to
$N=10^6$; for the counterexample pair it finds the period-four orbits, with
exact average $0.42487$. Typical $\beta=0$ orbits converge slowly (median
$\propto N^{-0.08}$), consistent with the $n^{-0.2}$ correlations. This is
numerical evidence for unique ergodicity under rational independence, not a
proof: a finite search cannot exclude very small exceptional sets.

\section{The asymmetric deformation: further data}
\label{sm:deformation}

\paragraph{Generic $\beta$.} The correlation at lag one vanishes for every
$\varepsilon$ (since $p_0$ is uniform given $q_0$), whereas at lag two it
vanishes only at $\varepsilon=1/2$. For $\varepsilon<0.04$,
$F=0.96$--$1.12$ (pair A) and $0.96$--$1.22$ (pair B), with statistical error
$\simeq0.015$; at $\varepsilon=0.1$--$0.2$ the tail of the $p$-modes already
deviates from the Poisson prediction, with a momentum relaxation time
$\sim1/\varepsilon$. The two ensembles of Fig.~\ref{fig:fano} differ by an
rms of $0.021$ in both pairs, against $0.022$ (A) and $0.021$ (B) expected, so
the structure is reproducible at every grid point; $F$ at $n/16$ is smaller
by $1.5\%$ on average ($1.8\%$ rms), and no extrapolation in $n$ is
attempted. Pair B has maxima of the smoothed $F$ near
$\varepsilon\simeq0.10,0.20,0.34,0.41$, $D$ up to $4.80$ at $\varepsilon=0.41$
and $D(1/2)=2.37$ (Experiment 12: $2.33$). Two zoomed windows,
$\varepsilon\in[0.28,0.30]$ ($400$ points) and $1/3\pm0.005$ ($401$ points),
each with two independent ensembles of $32\,768$ orbits at $n=2^{16}$, show a
smooth profile: the seeds agree within statistics (rms $z$-difference $1.00$
and $1.02$), polynomials of degree $8$--$12$ fit with $\chi^2/\text{dof}$
$1.04$--$1.17$, and $F(1/3)=2.50$ shows no feature at the rational. The
jaggedness of the wide grid is a sum of bumps of width $0.005$--$0.01$ plus
noise; no self-similar structure is seen down to a step $5\times10^{-5}$ at
noise $0.02$.

\paragraph{Tails at $\beta=0$, $\varepsilon=1/4$.} In two seeds of $10^5$
orbits to $2^{24}$, one orbit carries $3$--$18\%$ of $\sum X_n^2$, the kurtosis
is $10^2$--$4\times10^3$, the Hill index of $|X_n|$ is $2.0$--$2.3$ from
$n=2^{14}$ to $2^{24}$, and the raw $\Var(X_n)/n$ at $2^{24}$ differs between
seeds ($3.51$ and $5.17$). The interquartile scale
$\sigma_{\rm IQR}^2/n$, $\sigma_{\rm IQR}=(Q_{75}-Q_{25})/1.349$, equals
$1.34,1.89,2.26,2.66$ at $n=2^{10},2^{16},2^{20},2^{24}$ in both seeds, with
slope $0.133\pm0.001$ in $\ln n$; it is not the variance and is not compared
with $c(\varepsilon)$. The extreme orbits have uniform $(q_0,p_0)$, are
ordinary at earlier times (Spearman correlation of $|X|$ between $2^{16}$ and
$2^{24}$: $0.002$), and are all negative. The same run at $\varepsilon=1/2$
has light tails (kurtosis $4.7$--$5.3$, Hill index $5$--$9$) and
$\Var(S_{2^{24}})=94.8$, against $95.0$ from the fit $a+b\ln^2n$. In
Experiment 17 the stays of lengths $128$ to $2^{16}$ contribute $1.41$ to
$\mathbb E\,X^2/n$ at $\varepsilon=1/4$, against $1.44$ predicted, and the
law~\eqref{eq:trap-law} holds within $1$--$2\%$ up to $\ell\simeq2^{12}$.

\paragraph{The total trap coefficient (Experiment 19).}
Table~\ref{tab:ktot} gives $K_{\rm tot}(\varepsilon)$ for $\alpha=\sqrt2-1$
from all primitive zero-sum words of length $m\le60$. $K_{\rm tot}=16\varepsilon^2$
holds exactly for $\varepsilon\le0.2$ (only $m=2,4$); the first longer word,
$m=10$, appears at $\varepsilon\simeq0.2075$, then $m=8$ at $0.229$. $K_2$
saturates at $3\alpha^2=9-6\sqrt2$ for $0.3\le\varepsilon\le0.5$. The identity
$K_{\rm tot}=16\varepsilon^2$ holds for $\varepsilon\le0.05$ for five values
of $\alpha$ ($\sqrt2-1$, $1/e$, $\sqrt3-1$, $(\sqrt5-1)/2$, $\pi-3$), with an
$\alpha$-dependent range ($\pi-3$ deviates at $0.1$, the others at
$0.15$--$0.2$). For $\varepsilon\ge0.325$ the partial sums still grow at
$m=60$ and $c$ is only a lower bound; whether $K_{\rm tot}$ is finite there,
and how it behaves as $\varepsilon\to1/2$, is open.

\begin{table}[ht]
\centering
\caption{$K_{\rm tot}$ from zero-sum words of length $m\le30$ and $m\le60$,
and $c=(1-2\varepsilon)^2K_{\rm tot}(m\le60)$; $\alpha=\sqrt2-1$.}
\label{tab:ktot}
\begin{tabular}{cccccc}
\toprule
$\varepsilon$ & $K(m\le30)$ & $K(m\le60)$ & $16\varepsilon^2$ & $c$ & converged\\
\midrule
0.100 & 0.1600 & 0.1600 & 0.1600 & 0.1024 & yes\\
0.200 & 0.6400 & 0.6400 & 0.6400 & 0.2304 & yes\\
0.225 & 0.7855 & 0.7855 & 0.8100 & 0.2376 & yes\\
0.250 & 0.9760 & 0.9760 & 1.0000 & 0.2440 & yes\\
0.275 & 1.2579 & 1.2579 & 1.2100 & 0.2547 & yes\\
0.300 & 1.6889 & 1.6889 & 1.4400 & 0.2702 & yes\\
0.325 & 2.3836 & 2.4755 & 1.6900 & $\ge0.303$ & no\\
0.350 & 3.0907 & 3.5480 & 1.9600 & $\ge0.319$ & no\\
0.400 & 4.3435 & 5.5344 & 2.5600 & $\ge0.221$ & no\\
0.450 & 6.5154 & 11.4758 & 3.2400 & $\ge0.115$ & no\\
\bottomrule
\end{tabular}
\end{table}

\paragraph{Stay decomposition (Experiment 18).} Stays are maximal stretches
of period two ($s_n=s_{n-2}$, at least $66$ steps) or four (at least $132$
steps, without a period-two sub-stretch), $20\,000$ orbits to $2^{24}$ per
$\varepsilon$. Centered variances are essential, since outside the stays the
mean step is not zero. $\Var(X^{\rm rest}_n)/n$ goes $0.58\to0.61$
($\varepsilon=0.1$), $1.16\to1.28$ ($0.2$), $1.17\to1.33$ ($0.25$). At larger
$n$ the stay slope is lower than predicted ($0.099\pm0.051$ at
$\varepsilon=0.2$ on $[2^{18},2^{24}]$), as expected: in $2\times10^4$ orbits
fewer than one stay with $\ell\ge n/8$ is expected at $n=2^{24}$ ($\simeq0.02$
at $\varepsilon=0.2$), so the sample variance underestimates its mean; the
data support Eq.~\eqref{eq:clog} for $n\le2^{18}$ only, and $D_0$ is not
determined. With stays of all even periods up to $14$, at
$\varepsilon=0.25$, $0.3$, $0.35$, the slopes of $X^{\rm rest}$ are
$-0.001\pm0.002$, $0.002\pm0.002$ and $0.150\pm0.029$ (words longer than $14$
are missing at $0.35$), and the stay slope at $\varepsilon=0.3$ is
$0.15\pm0.05$ on $[2^{12},2^{18}]$, about $2\sigma$ below $c=0.270$. Hence
Eq.~\eqref{eq:clog} is not contradicted, but it is tested to about $40\%$
only for $\varepsilon>0.25$; at $\varepsilon=0.25$ the raw variance is too
noisy for a slope.

\paragraph{Zero drift and small $\beta$.} On the line
$\beta=-\alpha(1-2\varepsilon)$ at $\varepsilon=1/4$, the slope of
$\Var(X_n)/n$ in $\ln n$ goes from $-0.046$ ($2^8$--$2^{12}$) to
$-0.0023\pm0.0006$ ($2^{20}$--$2^{24}$); whether the limit is finite,
$\simeq0.23$, is not decided. At $\varepsilon=(\sqrt3-1)/3$ the same line
gives $1.47$ against $1.46$ at $\beta$ shifted by $10^{-3}$. For
$\beta=10^{-3}$ and $10^{-2}$ at $\varepsilon=1/4$ ($2048$ orbits,
$n\le2^{22}$) the variance is flat in $n$, the trap cutoff being
$\propto\beta^{-1/2}$.

\section{Phase-space exploration at finite resolution}
\label{sm:coverage}

To probe the exploration of the torus directly, we revisit the cell-coverage
test of Casati and Prosen~\cite{CasatiProsen2000}. Partition the torus into
$N=m^2$ equal square cells and let $n_N(t;z_0)$ be the number of distinct
cells visited by $z_0,\ldots,T^{t-1}z_0$; we estimate
$r_N(t)=\mathbb E_{z_0}[n_N(t;z_0)/N]$ from $4096$ independent uniform initial
conditions, counting coverage per orbit before averaging. For independent
uniform draws,
\begin{equation}
 r_N^{\mathrm{iid}}(t)=1-\left(1-\frac1N\right)^t
 \longrightarrow 1-e^{-s},\qquad s=\frac{t}{N},
 \label{eq:coverage-iid}
\end{equation}
and we use the exact finite-$N$ expression. We use the parameters of the
original study, which in our unit-torus convention read
$\alpha_{\rm cov}=(g-e^{-1})/4$, $\beta_{\rm cov}=(g+e^{-1})/4$, $g=(\sqrt5-1)/2$
(the change of coordinates $q=(x/2)\bmod1$, $p=(y/2)\bmod1$ from the torus
$[-1,1)^2$ halves both parameters and preserves the partition for even $m$),
and compare them with $(\alpha_{\rm cov},0)$. All coverage orbits are in
double precision; the $\beta=0$ data, to $t=1.6\times10^6$, are therefore
indicative only.

\paragraph{Generic coverage.} For $m=100,200,400,800$ and $t\le10N$ the
coverage nearly coincides with the independent reference
(Fig.~\ref{fig:coverage-generic}), with a small resolved deficit
$d_N(1)=r_N^{\rm iid}-r_N=(3.226,1.495,0.640,0.201)\times10^{-3}$ at $s=1$
(standard errors $(0.145,0.075,0.038,0.020)\times10^{-3}$; a repetition at
$m=400$ with another seed gives $0.677\pm0.039$). $\sqrt N\,d_N(1)$ decreases
from $0.323$ to $0.161$, so a simple $N^{-1/2}$ correction is not supported,
and no asymptotic exponent is assigned.

\paragraph{$\beta=0$.} Setting $\beta=0$ at the same $\alpha$ changes coverage
drastically (Fig.~\ref{fig:coverage-beta0}): at $t=10N$,
$r_N=0.11309\pm0.00054$, $0.06046\pm0.00028$, $0.03312\pm0.00017$ for
$m=100,200,400$, against $0.99995$ for independent draws; for $m=100$,
coverage rises only from $0.043$ to $0.081$ and $0.113$ for $t=10^3$, $10^4$,
$10^5$. Neither a logarithmic law nor a non-zero asymptotic unvisited
fraction is established. Ergodicity implies eventual visits to every cell
but no rate, and near-complete coverage does not imply ergodicity or
independence.

\begin{figure}[ht]
 \centering
 \includegraphics[width=\textwidth]{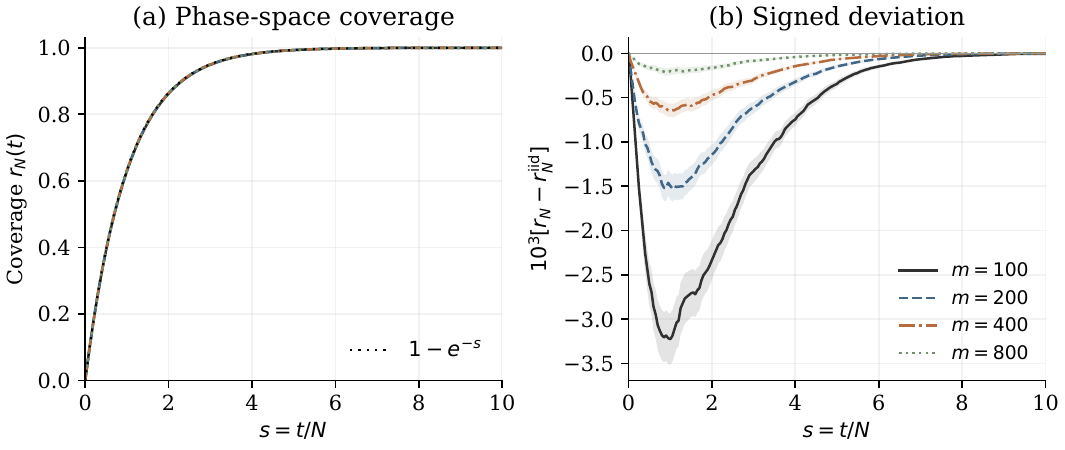}
 \caption{\textbf{Generic phase-space coverage}, $4096$ initial conditions per
 grid. (a) Mean covered fraction against $s=t/N$, with the exponential
 reference. (b) Signed deviation from the exact finite-$N$ independent
 reference, $\times10^3$; shading, two pointwise standard errors.}
 \label{fig:coverage-generic}
\end{figure}

\begin{figure}[ht]
 \centering
 \includegraphics[width=0.9\textwidth]{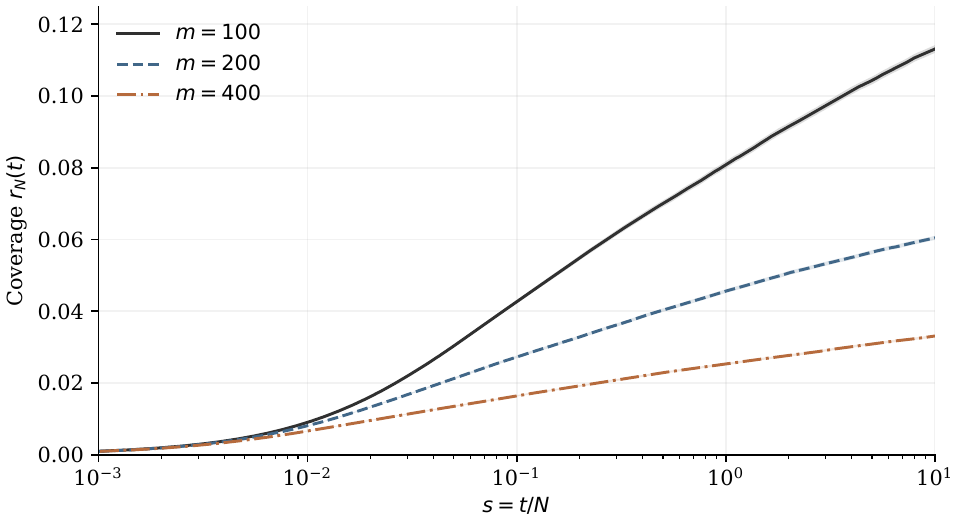}
 \caption{\textbf{Slow coverage for $\beta=0$} at $\alpha=\alpha_{\rm cov}$:
 coverage on a linear scale against logarithmic $s=t/N$ (shading: two
 pointwise standard errors); the independent reference, not shown, would be
 $\simeq99.995\%$ at $s=10$.}
 \label{fig:coverage-beta0}
\end{figure}

\fi

\end{document}